\documentclass[10pt,conference]{IEEEtran}
\pdfoutput=1
\RequirePackage{fix-cm}
\usepackage[T1]{fontenc}%
\usepackage[utf8]{inputenc}%
\usepackage[english]{babel}%

\usepackage{tabularx}%
\usepackage{etex}%
\usepackage{etoolbox}
\usepackage{xargs,ifthen}%
\usepackage{fixltx2e} %
\usepackage{fnpct}%
\usepackage%
[colorlinks,%
bookmarksnumbered=true,bookmarksopen=true,bookmarksdepth=2,
pdfauthor={Aubert, Bagnol, Seiller}]
{hyperref}

\usepackage{algorithmic}

\usepackage{tikz}
\usetikzlibrary{calc,arrows,matrix,positioning,fit,decorations.markings}

\usepackage{needspace}%
\usepackage{multicol}%
\usepackage{wrapfig}%
\usepackage{graphicx}%
\usepackage{enumitem}%
\setitemize{itemsep=0pt}
\setenumerate{itemsep=0pt}
\newlength{\savedparindent}
\setitemize{leftmargin=\savedparindent}
\setenumerate{leftmargin=\savedparindent}
\AfterEndEnvironment{itemize}{\noindentnext}
\AfterEndEnvironment{enumerate}{\noindentnext}

\usepackage{xcolor}%

\addto\extrasenglish{

}

\usepackage{blindtext}
\usepackage{xspace}%
\usepackage{microtype}%
\usepackage[autostyle = true, english = american, strict = true, autopunct=true]{csquotes}%
\usepackage{relsize}%
\usepackage{ellipsis}

\usepackage{amssymb}%
\usepackage{mathtools}%
\usepackage{amsthm}
\usepackage{thmtools}

\makeatletter
\RequirePackage{xkeyval}
\RequirePackage{ifthen}

\newdimen\ebproof@hsep
\newdimen\ebproof@rulemargin
\newdimen\ebproof@rulesep
\newdimen\ebproof@thickness
\newdimen\ebproof@labelsep
\newbox\ebproof@leftbox
\newbox\ebproof@rightbox

\define@key{ebproof}{hsep}{\setlength\ebproof@hsep{#1}}
\define@key{ebproof}{rulemargin}{\setlength\ebproof@rulemargin{#1}}
\define@key{ebproof}{rulesep}{\setlength\ebproof@rulesep{#1}}
\define@key{ebproof}{thickness}{\setlength\ebproof@thickness{#1}}
\define@key{ebproof}{labelsep}{\setlength\ebproof@labelsep{#1}}
\define@key{ebproof}{rulestyle}{%
  \expandafter\let\expandafter\ebproof@rule@do\csname ebproof@rule@#1\endcsname}
\define@boolkey{ebproof}{updown}[true]{}
\define@key{ebproof}{template}{\def\ebproof@template##1{#1}}
\define@key{ebproof}{lefttemplate}{\def\ebproof@lefttemplate##1{#1}}
\define@key{ebproof}{righttemplate}{\def\ebproof@righttemplate##1{#1}}
\define@key{ebproof}{leftlabeltemplate}{\def\ebproof@leftlabeltemplate##1{#1}}
\define@key{ebproof}{rightlabeltemplate}{\def\ebproof@rightlabeltemplate##1{#1}}
\define@key{ebproof}{leftlabel}{%
  \def\ebproof@setleftlabel{%
    \setbox\ebproof@leftbox=\hbox{\ebproof@leftlabeltemplate{#1}}}}
\define@key{ebproof}{rightlabel}{%
  \def\ebproof@setrightlabel{%
    \setbox\ebproof@rightbox=\hbox{\ebproof@rightlabeltemplate{#1}}}}
\define@boolkey{ebproof}{center}{}

\KV@ebproof@updownfalse
\def\ebproof@template#1{$#1$}
\def\ebproof@lefttemplate#1{$#1\mathrel{}$}
\def\ebproof@righttemplate#1{$\mathrel{}#1$}
\def\ebproof@leftlabeltemplate#1{#1}
\def\ebproof@rightlabeltemplate#1{#1}
\def\ebproof@setleftlabel{\setbox\ebproof@leftbox=\box\voidb@x}
\def\ebproof@setrightlabel{\setbox\ebproof@rightbox=\box\voidb@x}
\KV@ebproof@centertrue

\newcount\ebproof@localdimens
\def\ebproof@localdimen#1{%
  \advance\ebproof@localdimens1\relax
  \expandafter\dimendef\csname#1\endcsname\ebproof@localdimens
  \csname#1\endcsname=0pt\relax}

\newcount\ebproof@localboxes
\def\ebproof@localbox#1{%
  \advance\ebproof@localboxes1\relax
  \expandafter\chardef\csname#1\endcsname\ebproof@localboxes
  \setbox\csname#1\endcsname}

\def\ebproof@alloc#1{%
  \ebproof@localdimen{#1@left}%
  \ebproof@localdimen{#1@right}%
  \ebproof@localdimen{#1@axis}%
  \ebproof@localbox{#1@box}}

\newcount\ebproof@level
\newbox\ebproof@box@stack
\newbox\ebproof@box@temp

\def\ebproof@clear{%
  \global\ebproof@level=0%
  \global\setbox\ebproof@box@stack=\box\voidb@x%
  \gdef\ebproof@stack{}}

\def\ebproof@push#1{%
  \global\advance\ebproof@level1\relax
  \global\setbox\ebproof@box@stack=\hbox{%
    \unhbox\ebproof@box@stack\copy\csname#1@box\endcsname}%
  \xdef\ebproof@stack{%
    {\the\csname#1@left\endcsname}%
    {\the\csname#1@right\endcsname}%
    {\the\csname#1@axis\endcsname}%
    {\ebproof@stack}}}

\def\ebproof@pop#1{%
  \ifnum\ebproof@level>0\relax
    \global\advance\ebproof@level-1\relax
    \global\setbox\ebproof@box@stack=\hbox{%
      \unhbox\ebproof@box@stack
      \global\setbox\ebproof@box@temp=\lastbox}%
    \ebproof@alloc{#1}=\box\ebproof@box@temp%
    \begingroup\def\pop##1##2##3##4{\endgroup%
      \csname#1@left\endcsname=##1\relax
      \csname#1@right\endcsname=##2\relax
      \csname#1@axis\endcsname=##3\relax
      \gdef\ebproof@stack{##4}}%
    \expandafter\pop\ebproof@stack
  \else
    \PackageError{ebproof}{%
      Missing premiss in a proof tree}{}%
    \ebproof@alloc{#1}=\box\voidb@x%
  \fi}

\def\ebproof@pushsimple#1{%
  \begingroup
  \ebproof@alloc{A}=\hbox{#1}%
  \A@axis=.5\wd\A@box
  \ebproof@push{A}%
  \endgroup}

\def\ebproof@pushsplit#1#2{%
  \begingroup
  \ebproof@alloc{A}=\hbox{#1}%
  \A@axis=\wd\A@box
  \setbox\A@box=\hbox{\unhbox\A@box#2}%
  \ebproof@push{A}%
  \endgroup}

\def\ebproof@joinh{%
  \begingroup
  \ebproof@pop{A}%
  \ebproof@pop{B}%
  \ebproof@alloc{C}=\hbox{\box\B@box \kern\ebproof@hsep \box\A@box}%
  \C@left=\B@left
  \C@right=\A@right
  \C@axis=\wd\C@box
  \advance\C@axis\B@left
  \advance\C@axis-\A@right
  \divide\C@axis2\relax
  \ebproof@push{C}%
  \endgroup}

\def\ebproof@joinh@multi#1{%
  \begingroup
  \countdef\c=1
  \c=#1\relax%
  \ifnum\c=0
    \ebproof@alloc{X}=\hbox{}%
    \ebproof@push{X}%
  \else
    \ebproof@joinh@loop
  \fi
  \endgroup}
\def\ebproof@joinh@loop{%
  \ifnum\c>1
    \ebproof@joinh
    \advance\c-1
    \expandafter\ebproof@joinh@loop
  \fi}

\def\ebproof@joinright{%
  \begingroup
  \ebproof@pop{A}%
  \ebproof@pop{B}%
  \ebproof@alloc{C}=\hbox{\box\B@box \kern\ebproof@hsep \copy\A@box}%
  \C@left=\B@left
  \C@right=\B@right
  \advance\C@right\wd\A@box
  \advance\C@right\ebproof@hsep
  \ebproof@push{C}%
  \endgroup}

\def\ebproof@joinv{%
  \begingroup
  \ebproof@setleftlabel
  \ebproof@setrightlabel
  \ebproof@pop{A}%
  \ebproof@pop{B}%
  \ebproof@alloc{C}=\box\voidb@x%
  \ebproof@localdimen{A@shift}%
  \ebproof@localdimen{B@shift}%
  \ebproof@localdimen{R@shift}%
  \ebproof@localdimen{R@raise}%
  \ebproof@localdimen{R@width}%
  \ebproof@localdimen{C@width}%
  \ebproof@localdimen{tmp}%
  \ifdim\A@axis>\B@axis
    \A@shift=0pt%
    \B@shift=\A@axis
    \advance\B@shift-\B@axis
    \C@axis=\A@axis
  \else
    \A@shift=\B@axis
    \advance\A@shift-\A@axis
    \B@shift=0pt%
    \C@axis=\B@axis
  \fi
  \C@left=\A@left
  \advance\C@left\A@shift
  \C@right=\A@right
  \tmp=\wd\B@box
  \advance\tmp\B@shift
  \advance\tmp-\wd\A@box
  \advance\tmp-\A@shift
  \ifdim\tmp>0pt%
    \C@width=\wd\B@box
    \advance\C@width\B@shift
    \advance\C@right\tmp
  \else
    \C@width=\wd\A@box
    \advance\C@width\A@shift
  \fi
  \R@shift=\A@left
  \advance\R@shift\A@shift
  \tmp=\B@left
  \advance\tmp\B@shift
  \ifdim\R@shift>\tmp
    \R@shift=\tmp
  \fi
  \R@width=\wd\A@box
  \advance\R@width\A@shift
  \advance\R@width-\A@right
  \tmp=\wd\B@box
  \advance\tmp\B@shift
  \advance\tmp-\B@right
  \ifdim\tmp>\R@width
    \R@width=\tmp
  \fi
  \advance\R@width-\R@shift
  \ebproof@localbox{R@box}=\vbox{%
    \hsize=\R@width
    \hrule width \R@width height 0pt\relax
    \ebproof@rule@do}%
  \ifvoid\ebproof@leftbox\else
    \tmp=\wd\ebproof@leftbox
    \advance\tmp\ebproof@labelsep
    \ifdim\tmp>\R@shift
      \advance\tmp-\R@shift
      \advance\A@shift\tmp
      \advance\B@shift\tmp
      \advance\C@left\tmp
      \advance\C@axis\tmp
      \advance\C@width\tmp
      \R@shift=0pt\relax
    \else
      \advance\R@shift-\tmp
    \fi
  \fi
  \ebproof@localbox{RC@box}=\hbox{$\vcenter{\copy\R@box}$}%
  \R@raise=\ht\R@box
  \advance\R@raise-\ht\RC@box
  \setbox\RC@box=\hbox{%
    \ifvoid\ebproof@leftbox\else
      \copy\ebproof@leftbox
      \kern\ebproof@labelsep
    \fi
    \box\RC@box
    \ifvoid\ebproof@rightbox\else
      \kern\ebproof@labelsep
      \copy\ebproof@rightbox
    \fi}
  \tmp=\wd\RC@box
  \advance\tmp\R@shift
  \ifdim\tmp>\C@width
    \advance\tmp-\C@width
    \advance\C@right\tmp
  \fi
  \setbox\RC@box=\hbox{\raise\R@raise\box\RC@box}
  \ht\RC@box=\ht\R@box
  \dp\RC@box=\dp\R@box
  \ifKV@ebproof@updown
    \setbox\C@box=\vtop{%
      \moveright\A@shift\box\A@box
      \hrule height 0pt
      \moveright\R@shift\box\RC@box%
      \hrule height 0pt
      \moveright\B@shift\box\B@box}%
  \else
    \setbox\C@box=\vbox{%
      \moveright\B@shift\box\B@box
      \hrule height 0pt
      \moveright\R@shift\box\RC@box%
      \hrule height 0pt
      \moveright\A@shift\box\A@box}%
  \fi
  \ebproof@push{C}%
  \endgroup}
\def\ebproof@rule@none{%
  \kern\ebproof@rulemargin
}
\def\ebproof@rule@simple{%
  \kern\ebproof@rulemargin
  \hrule height \ebproof@thickness\relax
  \kern\ebproof@rulemargin
  }
\let\ebproof@rule@do=\ebproof@rule@simple
\def\ebproof@rule@double{%
  \kern\ebproof@rulemargin
  \hrule height \ebproof@thickness
  \kern\ebproof@rulesep
  \hrule height \ebproof@thickness
  \kern\ebproof@rulemargin}
\def\ebproof@rule@dashed{%
  \kern\ebproof@rulemargin%
  \hbox to \hsize{%
    \kern-1pt%
    \leaders\hbox{\kern1pt%
      \vrule height \ebproof@thickness width 5\ebproof@thickness%
      \kern1pt}\hfill
    \kern-1pt}%
  \kern\ebproof@rulemargin}

\def\ebproof@alter#1{%
  \begingroup
  \ebproof@pop{A}%
  \setbox\A@box=\hbox{{#1\box\A@box}}%
  \ebproof@push{A}%
  \endgroup}

\def\ebproof@delims#1#2{%
  \begingroup
  \ebproof@pop{TREE}%
  \ebproof@localbox{@SHIFTED}=%
    \hbox{$\vcenter{\copy\TREE@box}$}%
  \ebproof@localbox{@LEFT}=%
    \hbox{$#1\vrule height \ht\@SHIFTED depth \dp\@SHIFTED width 0pt\right.$}%
  \ebproof@localbox{@RIGHT}=%
    \hbox{$\left.\vrule height \ht\@SHIFTED depth \dp\@SHIFTED width 0pt#2$}%
  \ebproof@localdimen{dy}
  \dy=\dp\@SHIFTED
  \advance\dy-\dp\TREE@box
  \ebproof@alloc{A}=%
    \hbox{\raise\dy\hbox{\copy\@LEFT\box\@SHIFTED\copy\@RIGHT}}%
  \A@left=\wd\@LEFT \advance\A@left\TREE@left
  \A@right=\wd\@RIGHT \advance\A@right\TREE@right
  \A@axis=\wd\@LEFT \advance\A@axis\TREE@axis
  \ebproof@push{A}%
  \endgroup}

\def\ebproof@hypo@parse#1&#2&#3\ebproof@hypo@stop{
  \ifthenelse{\equal{#3}{}}%
    {\ebproof@pushsimple{\ebproof@template{#1}}}%
    {\ebproof@pushsplit
      {\ebproof@lefttemplate{#1}}%
      {\ebproof@righttemplate{#2}}}}

\newcommand\ebproof@hypo[2][]{%
  {\setkeys{ebproof}{#1}\ebproof@hypo@parse#2&&\ebproof@hypo@stop}}

\def\ebproof@infer{%
  \@ifnextchar[{\ebproof@infer@}{\ebproof@infer@[]}}
\def\ebproof@infer@[#1]#2{%
  \@ifnextchar[{\ebproof@infer@@{#1}{#2}}{\ebproof@infer@@{#1}{#2}[]}}
\def\ebproof@infer@@#1#2[#3]#4{%
  \ebproof@joinh@multi{#2}%
  \ebproof@hypo{#4}%
  {\setkeys{ebproof}{#1}%
   \ifthenelse{\equal{#3}{}}{}{\setkeys{ebproof}{rightlabel={#3}}}%
   \ebproof@joinv}}

\def\ebproof@ellipsis#1#2{%
  \ebproof@pushsimple{$\vdots$}%
  \ebproof@pushsimple{#1}%
  {\ebproof@hsep=\ebproof@labelsep\relax \ebproof@joinright}%
  {\let\ebproof@rule@do=\relax \ebproof@joinv}%
  \ebproof@hypo{#2}%
  {\let\ebproof@rule@do=\relax \ebproof@joinv}%
  }

\ebproof@clear

\def\ebproof@begin{%
  \edef\ebproof@start@level{\the\ebproof@level}%
  \setbox1=\vbox\bgroup
  \let\Hypo=\ebproof@hypo
  \let\Infer=\ebproof@infer
  \let\Ellipsis=\ebproof@ellipsis
  \let\Alter=\ebproof@alter
  \let\Delims=\ebproof@delims}
\def\ebproof@end{%
  \egroup
  \ebproof@pop{X}%
  \ifnum\ebproof@level=\ebproof@start@level\else
    \PackageError{ebproof}{Malformed proof tree}{%
      Some hypotheses were declared but not used in this tree.}%
  \fi
  \ifKV@ebproof@center
    \hbox{$\vcenter{\hbox{\box\X@box}}$}%
  \else
    \box\X@box
  \fi
  \global\setbox\ebproof@box@temp=\box1}

\newenvironment{prooftree*}[1][]{%
  \center
  \setkeys{ebproof}{#1}%
  \leavevmode\ebproof@begin
}{%
  \ebproof@end
  \endcenter}
\makeatother
\usepackage{cmll}%
\usepackage{stmaryrd}%

\AtBeginDocument{
\hyphenpenalty=1000
\tolerance=400
\setlength{\parskip}{2pt}
}

\usepackage{flushend} %

\makeatletter%
\newcommand{\noindentnext}{\par\@afterindentfalse\@afterheading} %
\makeatother%

\let\proof\undefined%
\let\endproof\undefined%

\declaretheoremstyle[
spaceabove=0.5\baselineskip, spacebelow=0.5\baselineskip,
headfont=\normalfont\bfseries,
notefont=\normalfont\bfseries,
headpunct=.,
bodyfont=\itshape,
preheadhook=,%
postheadhook=,%
prefoothook=,%
postfoothook=\aftergroup\noindentnext%
]%
{general}

\declaretheoremstyle[
spaceabove=0.5\baselineskip, spacebelow=0.5\baselineskip,
headfont=\normalfont\itshape,
notefont=\normalfont\itshape,
bodyfont=\normalfont,
]
{remark}

\declaretheoremstyle[
spaceabove=0.5\baselineskip, spacebelow=0.5\baselineskip,
headfont=\normalfont\bf,
notefont=\normalfont\bf,
bodyfont=\normalfont,
]
{notation}

\declaretheoremstyle[
spaceabove=0.5\baselineskip, spacebelow=0.5\baselineskip,
headpunct=~\raisebox{0.5pt}{$\mathsmaller\blacktriangleright\,$},
headfont=\normalfont\itshape,
notefont=\normalfont\itshape,
bodyfont=\normalfont,
qed=\raisebox{0.5pt}{$\mathsmaller\blacktriangleleft$}
]
{proof}

\declaretheoremstyle[
spaceabove=0.5\baselineskip, spacebelow=0.5\baselineskip,
headpunct=~\raisebox{0.5pt}{$\mathsmaller\blacktriangleright\,$},
headfont=\normalfont\itshape,
notefont=\normalfont\itshape,
bodyfont=\normalfont,
qed=\raisebox{0.5pt}{}
]
{manual_qed_proof}

\newcommand{\mqed}{\hfill\raisebox{0.5pt}{$\mathsmaller\blacktriangleleft$}}

\declaretheorem[style=general,numberwithin=section]{definition}
\declaretheorem[style=general,sibling=definition]{theorem}
\declaretheorem[style=general,sibling=definition]{proposition}
\declaretheorem[style=general,sibling=definition]{corollary}
\declaretheorem[style=general,sibling=definition]{lemma}

\declaretheorem[style=general,sibling=definition]{notation}

\declaretheorem[style=remark,sibling=definition]{remark}
\declaretheorem[style=remark,sibling=definition]{example}

\declaretheorem[style=proof,numbered=no]{proof}
\declaretheorem[style=manual_qed_proof,numbered=no,name=Proof]{manual_qed_proof}
\newcommand{\term}[1]{\textnormal{#1}} %
\newcommand{\rep}[1]{\bar{#1}}
\newcommand{\letter}[1]{\symbf #1}

\newcommand{\simpleoring}[1]{{\algfont O}[\,#1\,]} %
\newcommand{\mnfas}[1]{\(\mathrm{2MFA\textnormal{+}S(#1)}\)}

\newcommand{\sansparammnfas}{\(\mathrm{2MFA\textnormal{+}S}\)\xspace}
\newcommand{\sansparamboldmnfas}{\(\mathrm{\mathbf{2MFA}\textnormal{\textbf{+}}\mathbf{S}}\)\xspace}

\newcommand{\modulo}{\emph{modulo}\xspace}
\newcommand{\cut}{cut\xspace}

\newcommand{\GoI}{{\sc GoI}\xspace}

\newcommand{\GofI}{geometry of interaction\xspace}

\newcommand{\locit}[1]{#1} %
\newcommand{\eg}{\locit{e.g.}~}
\newcommand{\ie}{\locit{i.e.}~}
\newcommand{\via}{\locit{via}~}

\renewcommand{\iff}{\locit{iff}~}

\newcommand{\resp}{resp.~}
\newcommand{\etc}{\locit{etc.}\xspace}
\newcommand{\etal}{\locit{et al.}\xspace}

\newcommand{\incise}[1]{---\,#1\,---}
\newcommand{\demiincise}[1]{---\,#1}

\newcommand{\compclass}[1]{\textsc{#1}}%
\newcommand{\obs}[1]{{#1}}%
\newcommand{\state}[1]{\textbf{#1}\xspace}%
\newcommand{\langage}[1]{\mathcal{#1}}%
\newcommand{\problem}[1]{\textsc{#1}}%
\newcommand{\symbf}{\mathtt}%
\newcommand{\algfont}{\mathcal}%
\newcommand{\positions}[1]{\vec{\mathtt{#1}}}%
\newcommand{\var}[1]{\mathit #1}%
\newcommand{\func}[1]{\underline{\mathtt{#1}}}%
\newcommand{\functerm}[1]{\underline{\term{#1}}}%
\newcommand{\queryf}[1]{\mathbf{#1}}

\newcommand{\NN}{\mathbb N}%
\newcommand{\sequ}[1]{\vec{#1}}%
\newcommand{\powerset}{\mathcal P}%
\newcommand{\eqdef}{:=}%
\DeclareMathOperator{\card}{Card}
\newcommand{\extlist}[2]{#1,\,\dots\,,#2}%
\newcommand{\extinflist}[1]{#1,\,\dots}%
\newcommandx{\extset}[3][1=]{#1\{\,#2,\,\dots\,,#3\,#1\}}%
\newcommandx{\set}[3][1=]{#1\{\:#2 \ #1|\ #3\:#1\}}%
\newcommand{\void}{\varnothing}%
\newcommand{\unit}{I}%

\newcommand{\pr}{\texttt{\textit{p}}}%
\newcommand{\isnilp}[1]{\mathbf{Nil}(#1)} %

\newcommand{\Nat}{\texttt{Nat}}
\renewcommand{\with}{\binampersand}

\renewcommand{\wn}{\textbf{?}}

\newcommand{\Logspace}{\compclass{Logspace}\xspace}%
\newcommand{\DLogspace}{\compclass{DLogspace}\xspace}%
\newcommand{\NLogspace}{\compclass{NLogspace}\xspace}%
\newcommand{\coNLogspace}{\compclass{co-NLog\-space}\xspace}%
\newcommand{\Ptime}{\compclass{Ptime}\xspace}%

\newcommand{\FPtime}{\compclass{FPtime}\xspace}%

\newcommand{\flow}{\leftharpoonup}%
\newcommand{\sflow}{\leftrightharpoons}%
\newcommand{\wordterm}[3]{\letter #1\p #2 \p \var x \p \var y \p \head(\symbf #3)}%
\newcommand{\wordrep}[2]{\rep{#1}_#2}%
\newcommand{\dummy}{\star}%
\newcommand{\op}[2]{\texttt{\textsc{op}}_{#1,\,#2}}%
\newcommand{\trunc}[1]{T_{#1}}%
\newcommand{\inc}[1]{#1^\uparrow}%
\newcommand{\dec}[1]{#1^\downarrow}%

\newcommand{\sat}{\texttt{\textit{satur}}}%
\newcommand{\short}{\texttt{\textit{short}}}%

\newcommand{\alphabet}{\Sigma_{\lft\rgt}}
\newcommand{\stack}{\algfont S\mathit{tack}}%
\newcommand{\balanced}{\algfont R_{\mathbf b}}%
\newcommand{\Res}{\algfont R}%
\newcommand{\usring}{\Res}%
\newcommand{\unused}[1]{{\symbf #1}^{\perp}} %
\newcommand{\stackobs}{\algfont O^{\mathbf{b+s}}}%

\newcommand{\denc}[1]{[#1]}
\newcommand{\genc}[1]{\langle#1\rangle}

\newcommand{\Var}{\texttt{{var}}}%
\newcommand{\uequ}{\,\texttt{\smaller=}\,}%
\newcommand{\lang}[1]{\langage{L}(#1)}%

\newcommand{\lft}{\mathtt l}
\newcommand{\rgt}{\mathtt r}
\newcommand{\LR}{\{\lft,\rgt\}}
\newcommand{\size}[1]{|#1|}
\newcommand{\head}{\texttt{\textsc{h\!e\!a\!d}}} %
\newcommand{\auxp}{\texttt{\textsc{a\!u\!x}}} %
\newcommand{\h}{\texttt{\textit{h}}} %
\newcommand{\start}{\star}
\newcommand{\p}{\!\mathrel{\textstyle\mathsmaller\bullet}\!}
\newcommand{\ptext}{$\mathrel{\mathsmaller\bullet}$\xspace} %
\newcommand{\algo}[2]{\textsc{#1}_{#2}}%
\newcommand{\automate}[1]{#1}

\newcommand{\org}{\texttt{or}\xspace}
\newcommand{\andg}{\texttt{and}\xspace}
\newcommand{\notg}{\texttt{not}\xspace}
\newcommand{\zerog}{\texttt{0}\xspace}
\newcommand{\oneg}{\texttt{1}\xspace}

\newcommand{\gate}[3]{#1\rhd^{#2}#3}

\newcommand{\ore}[3]{\gate{#1,#2}{\org}{#3}}
\newcommand{\ande}[3]{\gate{#1,#2}{\andg}{#3}}
\newcommand{\note}[2]{\gate{#1}{\notg}{#2}}
\newcommand{\zeroe}[1]{\gate{}{\zerog}{#1}}
\newcommand{\onee}[1]{\gate{}{\oneg}{#1}}

\newcommand{\funv}[1]{\mathtt #1}
\newcommand{\funvn}[1]{\overline{\mathtt #1}}

\newcommand{\enc}[1]{[#1]}

\newcommand{\bstack}{\boxdot}

\newcommand{\tridro}{\vartriangleright}
\newcommand{\trigau}{\vartriangleleft}

\newcommand{\onstack}[1]{\texttt{#1}\xspace}
\newcommand{\pop}{\onstack{pop}}
\newcommand{\peek}{\onstack{peek}}
\newcommand{\push}{\onstack{push}}

\newcommand{\Mbis}{\head}
\newcommand{\nary}[1]{\auxp_{#1}}

\newcommand{\flatpush}[1]{\mathtt l^{\texttt{push}}_{#1}}
\newcommand{\flatpop}[1]{\mathtt l^{\texttt{pop}}_{#1}}

\IEEEoverridecommandlockouts %

\newtoggle{newpage} %
\togglefalse{newpage}
\newcommand{\sepsection}{\iftoggle{newpage}{\newpage}{}}

\begin{document}

\title{Memoization for Unary Logic Programming: Characterizing \Ptime}

\author{
\IEEEauthorblockN{Clément Aubert}
\IEEEauthorblockA{Inria\\Université Paris-Est, LACL (EA 4219), UPEC,\\F-94010 Créteil, France}
\and
\IEEEauthorblockN{Marc Bagnol}
\IEEEauthorblockA{Aix Marseille Université,\\CNRS, Centrale Marseille, I2M UMR 7373,\\13453, Marseille, France}
\and
\IEEEauthorblockN{Thomas Seiller}
\IEEEauthorblockA{IHÉS}
\thanks{This work was partly supported by the ANR-10-BLAN-0213 Logoi, the ANR-11- BS02-0010 Récré and the ANR-11-INSE-0007 REVER.}
}

\maketitle

\begin{abstract}
We give a characterization of deterministic polynomial time
computation based on an algebraic structure called the resolution
semiring, whose elements can be understood as logic programs or sets of rewriting rules
over first-order terms.

More precisely, we study the restriction of this framework to terms
(and logic programs, rewriting rules) using only unary symbols. We prove it is complete
for polynomial time computation, using an encoding of pushdown
automata. We then introduce an algebraic counterpart of the
memoization technique in order to show its \Ptime soundness.

We finally relate our approach and complexity
results to complexity of logic programming. As an application of our techniques, we
show a \Ptime-completeness result for a class of logic programming
queries which use only unary function symbols.

\end{abstract}

\begin{IEEEkeywords}
Implicit Complexity,
Resolution,
Logic Programming,
Polynomial Time,
Proof Theory,
Pushdown Automata,
Geometry of Interaction.
\end{IEEEkeywords}

\section{Introduction}
Complexity theory focuses on questions related to resource usage of computer programs,
such as the amount of time or memory a given program will need to solve a problem.%

Complexity classes are defined as sets of problems that can be solved by algorithms
whose executions
need comparable amounts of resources. For instance, the class \Ptime is the set of predicates
over binary words that can be decided by a Turing machine implementing an algorithm whose execution time is bounded
by a polynomial in the size of its input.

However, these definitions depend on the notion of machine %
and cost-model considered, for the efficiency of an algorithm is sensible to these. %
The \enquote{invariance thesis}~\cite{Boas1990} is a way to bypass this limitation by defining what \enquote{a reasonable model} is: all the \enquote{reasonable} models (endowed with cost models) can simulate each other with a \enquote{reasonable} overhead.
The bootstrap for this notion to apply largely was to remark that polynomial bounds on execution time are robust,
as the class of problems captured by different models where this bound coincide.
The definition is still machine-dependent, \emph{but not dependent of a particular model of computation}.

One of the main motivations for an implicit computational complexity (ICC) theory is to find completely machine-independent characterizations of complexity classes.
The aim is to characterize classes not \enquote{\it by constraining the amount of resources a machine is allowed to use, but rather by imposing linguistic constraints on the way algorithms are formulated.} \cite[p.~90]{DalLago2012a}
This has been already achieved \via different approaches, one of which is based on considering
restricted programming languages or computational principles~\cite{Bellantoni1992a,Leivant1993,Neergaard2004}.

A number of results also arose from proof theory through the study of subsystems of linear logic~\cite{Girard1987}.
More precisely, the Curry-Howard \incise{or \emph{proofs as programs}} correspondence expresses a deep
relation between formal proofs and typed programs. For instance, one can define a formula $\Nat$ which
corresponds to the type of binary integers, in the sense that a given (\cut-free, \ie normal, already evaluated) proof of this type
represents a given natural number.
A proof of the formula $\Nat\Rightarrow \Nat$ then corresponds to an algorithm computing
a function from integers to integers, where the computation
itself amounts to a rewriting on proofs: the \cut-elimination procedure.

By restricting the rules of the logical system, one obtains a subsystem where
\emph{less} proofs of type $\Nat\Rightarrow \Nat$ can be written,
hence \emph{less} algorithms can be represented. In a number of such restricted systems
the class of accepted proofs, \ie of programs, corresponds%
\footnote{We mean \emph{extensional} correspondence: they compute the same functions.}
to some complexity class:
elementary complexity~\cite{Girard1995,Danos2003}, polynomial time~\cite{Lafont2004,Baillot2004},
logarithmic~\cite{Lago2010d} and polynomial~\cite{Gaboardi2012} space.

More recently, new methods for obtaining implicit characterizations of complexity classes
based on the \emph{\GofI} (\GoI) research program~\cite{Girard1989b} have been developed.
The \GoI approach offers a more abstract and algebraic point of view
on the \cut-elimination procedure of linear logic.
One works with a set of \emph{untyped programs} represented as some geometric objects,
\eg graphs \cite{Danos1990,Seiller2012} or generalizations of graphs \cite{Seiller2014a},
bounded linear maps between Hilbert spaces (operators) \cite{Girard1989a,Girard2011a,Seiller2014b},
clauses (or \enquote{flows}) \cite{Girard1995a,Bagnol2014}.
This set of objects is then considered together with an abstract notion of execution, seen as an interactive process:
a function does not process a static input,
but rather communicate with it, asking for values, reading its answers, asking for another value, \etc

Types can then be defined as sets of program representations sharing comparable behaviors.
For instance the type $\Nat \Rightarrow\Nat$ is the set of untyped programs which, given an integer as
input, produce an integer as output.

This approach based on the \GoI differs from previous ICC works using linear logic in that they do not
rely on a restriction of some type system, but rather on a restriction on the set of
program representations considered. Still, they benefit from previous works in type theory:
for instance the representation of integers used here comes from their representation
in linear logic, translated in the \GoI setting, whose interactive point of view
on computation has proven crucial in characterizing logarithmic space computation~\cite{Lago2010d}.

The first results that used those innovative considerations were based on operator algebras~\cite{Girard2012,Aubert2014ctemp,Aubert2015temp}.
Here we consider a more syntactic
flavor of the \GoI where untyped programs are represented in the so-called
\emph{resolution semiring}~\cite{Bagnol2014}, a semiring based on the resolution rule~\cite{Robinson1965}
and a specific class of logic programs.
This setting presents some advantages: it avoids the involvement of operator algebras theory,
it eases the discussions in terms of complexity (we manipulate first-order terms, which have natural
notions of size, height, \etc) and it offers a straightforward connection with complexity of
logic programming~\cite{Dantsin2001}.

Previous works in this direction led to characterizations of logarithmic space predicates \Logspace and
\coNLogspace~\cite{Aubert2014,Aubert2014b}, by considering for instance restrictions on the height of variables.

Our main contribution here is a characterization of the class \Ptime by studying a natural restriction,
namely that one is allowed to use exclusively unary function symbols.
Pushdown automata%
\footnote{More precisely, \(2\)-way \(k\)-head non-deterministic finite automata with pushdown stack. See \autoref{def-automata}.}
are easily related to this simple restriction, for they
can be represented as logical programs satisfying
this \enquote{unarity} restriction. This will imply the completeness of the model under consideration
for polynomial time predicates.

We then complete the characterization by showing that any such unary logic program can be decided
in polynomial time. This part of the proof consists in an adaptation of S.~Cook's memoization
technique~\cite{Cook1971a} to the context of logic programs.

The last part of the paper presents consequences of these results in terms of complexity of
logic programming, namely that the corresponding class of queries are \Ptime-complete, when considering
combined complexity~\cite[p.~380]{Dantsin2001}.

Compared to other ICC characterizations of \Ptime, and in particular those coming from proof theory,
our results have a simple formulation
and provide an original point of view on complexity classes.

A byproduct of this work is to provide a %
method to test membership in \Ptime:
if one can rephrase a problem with clauses \(H \dashv B\) using only unary function symbols,
then our result %
ensures that the problem lies in \Ptime.
Conversely if a problem cannot be rephrased that way, it %
lies outside of \Ptime.

\subsection{Outline of the paper}

We begin by giving in \autoref{semiring} the formal definition of the resolution semiring;
then briefly explain how words can be represented in this structure (\autoref{sec_rep})
and recall the characterization of logarithmic space obtained in earlier work (\autoref{sec_logspace}).
In \autoref{sec_stack} we introduce the restricted semiring that will be under study in this paper:
the $\stack$ semiring.

The next two sections are respectively devoted to the completeness and soundness results for \Ptime.
For completeness, we first review the fact that multi-head finite automata with pushdown stack
characterize \Ptime and review the memoization technique in this case (\autoref{sec_cook}), and then show how to represent them as elements built from the $\stack$ semiring (\autoref{sec_completeness}).
The soundness result is then obtained by adapting memoization to the $\stack$ semiring.
This adaptation, which we call the \emph{saturation} technique, is introduced in \autoref{saturation-section}.

In the last section, we formulate our results in terms of complexity of logic programming.
In particular, we explain how elements of the $\stack$ semiring can be seen as
a particular kind of unary logic programs to which the saturation technique can be applied.
This allows us to show that the combined complexity problem for unary logic program is \Ptime-complete.

As an illustration%
, we show in \autoref{sec_cvp} that the circuit value problem can be solved with
this method.

\sepsection
\section{The Resolution Semiring}
\subsection{Flows and Wirings}\label{semiring}
Let us begin with some reminders and notations on first-order terms and unification theory.

\begin{notation}[terms]\label{not_terms}
We consider first-order terms, written $\extinflist{\term{t}, \term{u}, \term{v}}$, built from variables
and function symbols with assigned finite arity. Symbols of arity $0$ will be called \emph{constants}.

Sets of variables and of function symbols of any arity are supposed infinite. Variables will be noted in italics font (\eg $x,y$) and function symbols in typewriter font
(\eg $\symbf c, \symbf f(\cdot),\symbf g(\cdot,\cdot)$).

We distinguish a binary function symbol \ptext (in infix notation) and a constant symbol $\start$.
We will omit the parentheses for \ptext and write \(\term{t}\p \term{u}\p \term{v}\) for \(\term{t}\p(\term{u}\p \term{v})\).

We write $\Var(\term{t})$ the set of variables occurring in the term $\term{t}$ and say that $\term{t}$ is \emph{closed} if $\Var(\term{t})=\void$.
The \emph{height} $\h(\term{t})$ of a term $\term{t}$ is the maximal distance between its root and leaves; a variable occurrence's height in \(\term{t}\) is its distance to the root. \label{def-height}

We will write $\theta \term{t}$ the result of applying the substitution $\theta$ to the term $\term t$
and will call \emph{renaming} a substitution $\alpha$ that bijectively maps variables to variables.
\end{notation}

We will be concerned with formal solving of equations of the form $\term t\uequ\term u$ where
$\term t$ and $\term u$ are terms.
Let us introduce a precise formulation of this problem and some associated vocabulary.

\begin{definition}[unification, matching and disjointness]\label{disjoint}
Two terms $\term{t}, \term{u}$ are:
\begin{itemize}
[nosep,
,noitemsep]
\item \emph{unifiable} if there exists a substitution $\theta$
\incise{a \emph{unifier} of $\term{t}$ and $\term{u}$}
such that \( \theta \term{t} = \theta \term{u}\).
If any other unifier of $\term{t}$ and $\term{u}$ is an instance of $\theta$,
we say $\theta$ is the \emph{most general unifier} (MGU) of $\term{t}$ and $\term{u}$;
\item \emph{matchable} if $\term t',\term u'$ are unifiable,
where $\term t',\term u'$ are renamings
of $\term{t},\term{u}$ such that $\Var(\term t') \cap \Var(\term u') = \varnothing$;
\item \emph{disjoint} if they are not matchable.
\end{itemize}
\end{definition}

A fundamental result of unification theory is that when two terms are unifiable,
a MGU exists and is computable.
More specifically, the problem of deciding whether two terms are unifiable is $\Ptime$-complete
\cite[Theorem 1]{Dwork1984}.

The notion of MGU allows to formulate the \emph{resolution rule}, a key concept of
logic programming that defines the composition of Horn clauses
(expressions of the form $H \dashv \extlist{B_1}{B_n}$):

\begin{prooftree*}
\Hypo{V &\dashv \extlist{T_1}{T_n}}
\Infer[rulestyle=none]{1}{H &\dashv \extlist{B_1}{B_m},U}
\Hypo{\Var(U)\cap\Var(V)=\varnothing}                                      %
\Infer[rulestyle=none]{1}{\theta\text{ is a MGU of } U \text{ and }V}
\Infer{2}[\tt Res]{\theta H \dashv \extlist{\theta B_1}{\theta B_m},\extlist{\theta T_1}{\theta T_n}}
\end{prooftree*}

Note that the condition on variables implies that we are matching $U$ and $V$ rather than unifying them.
In other words, the resolution rule deals with variables as if they were bounded.

From this perspective, \enquote{flows} \incise{defined below} are a specific type of Horn clauses $H \dashv B$, with exactly one formula $B$ on the right of $\dashv$ and all the variables of $H$ already appearing in $B$. The product of flows will be defined as the resolution rule restricted to this specific type of clauses.

\begin{definition}[flow]\label{def_flow}
A \emph{flow} is an ordered pair $f$ of terms $f\eqdef \term{t}\flow \term{u}$,
with $\Var(\term{t})\subseteq\Var(\term{u})$.
Flows are considered up to renaming: for any renaming $\alpha$,
$\term{t} \flow \term{u}\,=\,\alpha \term{t} \flow \alpha \term{u}$.
\end{definition}

A flow $\term{t}\flow\term{u}$ can also be understood as a rewriting rule over the set of first-order terms. For instance, the flow $\symbf{g}(x)\flow\symbf{f}(x)$ corresponds to the following rewriting rule: terms of the form $\symbf{f}(\term{v})$ where $\term{v}$ is a term are rewritten as $\symbf{g}(\term{v})$ and all other terms are left unchanged.

We will soon define %
the \emph{product} of flows which provides a way of composing them;
from the term-rewriting perspective, this operation corresponds to composing two rules
\incise{when possible, \ie when the result of the first rewriting rule allows the
application of the second} into a single one.

For instance, one can compose the flows $f_{1}\eqdef\symbf{h}(x)\flow\symbf{g}(x)$ and $f_{2}\eqdef\symbf{g}(x)\flow\symbf{f}(x)$ to produce the flow $f_{1}f_{2}=\symbf{h}(x)\flow\symbf{f}(x)$. Notice by the way that this (partial) product is not commutative as composing these rules the other way around is impossible, \ie $f_{2}f_{1}$ is not defined.

\begin{definition}[product of flows]
Let $\term{t}\flow \term{u}$ and $\term{v}\flow \term{w}$ be two flows.
Suppose we picked representatives of the renaming classes
such that $\Var(\term{u})\cap\Var(\term{v})=\varnothing$.

The \emph{product} of $\term{t}\flow \term{u}$ and $\term{v}\flow \term{w}$ is defined
when $\term{u}$ and $\term{v}$ are unifiable, with MGU $\theta$, as
$(\term{t}\flow \term{u})(\term{v}\flow \term{w})\,\eqdef \:\theta \term{t}\flow \theta \term{w}$.
\end{definition}

We now define wirings, which are simply finite sets of flows and therefore correspond to logic programs.
From the term-rewriting perspective they are just sets of rewriting rules.
The definition of product of flows is naturally lifted to wirings.

\begin{definition}[wiring]\label{def_wirings}
A \emph{wiring} is a finite set of flows.
Their product is defined as
$FG\eqdef \set{fg}{f \in F, \, g \in G, \: fg \text{ defined}}$.
The \emph{resolution semiring} $\usring$ is the set of all wirings.
\end{definition}

The set of wirings $\usring$ indeed enjoys a structure of semiring%
\footnote{A \emph{semiring} is a set $R$ equipped with two operations $+$ (the sum) and $\times$ (the product, whose symbol is usually omitted), and an element $0\in R$ such that: $(R,+,0)$ is a commutative monoid;
$(R,\times)$ is a \emph{semigroup}, \ie a monoid which may not have a neutral element; the product distributes over the sum; the element $0$ is absorbent: $0r=r0=0$ for all $r\in R$.}.
We will use an \emph{additive notation} for sets of flows to highlight this situation:

\begin{itemize}
\item The symbol $+$ will be used in place of $\cup$, and we write sets as sums of their elements: $\extset{f_1}{f_n}\eqdef f_1 + \cdots + f_n$.
\item We denote by $0$ the empty set, \ie the unit of $+$.
\item We have a unit for the product, the wiring $\unit \eqdef x \flow x$.
\end{itemize}

As we will always be working within $\usring$, the term \enquote{semiring}
will be used instead of \enquote{subsemiring of $\usring$}.

Finally, let us recall the notion of nilpotency in a semiring and extend the notion of height (of terms) to flows and wirings.

\begin{definition}[height]\label{def_height}
The \emph{height} $\h(f)$ of a flow $f=\term{t} \flow \term{u}$ is defined as $\max\{\h(\term{t}),\h(\term{u})\}$.
A wiring's \emph{height} is defined as $\h(F)=\max\set{h(f)}{f\in F}$.
By convention $\h(0)=0$.
\end{definition}

\begin{definition}[nilpotency]\label{def_nilp}
A wiring ${F}$ is \emph{nilpotent} \incise{written $\isnilp{F}$} if and only if ${F}^n=0$ for some $n$.
\end{definition}

The above classical notion from abstract algebra has a specific reading in our case of study.
In terms of logic
programming, it means that all chains obtained by applying  the resolution rule to the set of clauses
we consider cannot be longer than a certain bound.
From the point of view of rewriting, it means that the
set of rewriting rules we consider is terminating
with a uniform bound on the length of rewriting chains
\demiincise{note however that we consider rewriting that occur only at the root of terms, while the usual
notion from term rewriting systems~\cite{Baader1998a} allows in-context rewriting}.

\subsection{Representation of Words and Programs}\label{sec_rep}

This section explains and motivates the representation of words as flows.
By studying their interactions with wirings from a specific semiring, notions of program and language are defined.

First, let us see how the binary function symbol $\p$ used to construct terms can be extended to build flows and then
semirings.

\begin{definition}\label{ptensor}
Let $\,\term{u}\flow \term{v}\,$ and $\,\term{t}\flow \term{w}\,$ be two flows.
Suppose we have chosen representatives of their renaming classes that have disjoint sets of variables.

We define
$(\term{u}\flow \term{v}) \p (\term{t}\flow \term{w}):=\: \term{u}\p \term{t} \flow \term{v}\p \term{w}$.
The operation is extended to wirings by $(\sum_i f_i)\p(\sum_j g_j):= \sum_{i,j} f_i\p g_j$.

Then, given two semirings $\algfont A$ and $\algfont B$, we define the semiring %
$\algfont A\p \algfont B:= \,\set{\sum_i F_i\p G_i}{F_i\in\algfont A\,,\,G_i\in \algfont B}$.
\end{definition}

The operation indeed defines a semiring because for any wirings $F,F',G,G'$
we have $(F\p G)(F\p G)=FF'\p GG'$.
Moreover, we carry on the convention of writing $\algfont A \p \algfont B \p \algfont C$
for $\algfont A \p (\algfont B \p \algfont C)$.

\begin{notation}
We write $\term{t}\sflow\term{u}$ the sum $\term{t}\flow\term{u}+\term{u}\flow\term{t}$.
\end{notation}

\begin{definition}[word representations]\label{word-rep}
From now on, we suppose fixed an infinite set of constant symbols $\symbf P$ (the \emph{position constants})
and a finite alphabet $\Sigma$ disjoint from $\symbf P$ with $\dummy\not\in\Sigma$
(we write $\Sigma^\ast$ the set of words over $\Sigma$).

Let $W = \symbf{c}_1\cdots \symbf{c}_n\in\Sigma^\ast$ and $p=\extlist{\symbf p_0,\symbf p_1}{\symbf p_n}$ be pairwise distinct elements of $\symbf P$.

Writing $\symbf p_{n+1}=\symbf p_{0}$ and $\symbf{c}_{n+1} = \symbf{c}_{0} = \start$,
we define the \emph{representation} of $W$ associated with $\extlist{\symbf p_0,\symbf p_1}{\symbf p_n}$
as the following wiring: %
\[
\rep{W}_{p}= \sum_{i=0}^{n}{\wordterm{c_{i}}{\rgt}{p_{i}}}
\sflow {\wordterm{c_{i+1}}{\lft}{p_{i+1}}}
\]
\end{definition}

In this definition, the position constants represent memory cells storing the symbols
$\star$, $\symbf{c}_{1}$, $\symbf{c}_{2}$, \dots.

The representation of words is \emph{dynamic}, \ie we may think intuitively of
\emph{movement instructions} from a symbol to the next or the previous (hence the choice of symbols
$\lft$ and $\rgt$ for \enquote{left/previous} and \enquote{right/next}) for some kind of
automaton reading the input. More details on this will be given in the proof of \autoref{th_pcomplete}.

Hence, for a given position constant $\symbf{p}_{i}$, we use terms $\symbf{c}_{i}\p \rgt$
and $\symbf{c}_{i}\p \lft$ which will be linked (by flows of the representation) to elements
$\symbf{c}_{i+1}\p \lft$ at position $\symbf{p}_{i+1}$ and $\symbf{c}_{i-1}\p \rgt$ at position
$\symbf{p}_{i-1}$ respectively.

Note moreover that the %
representation of the input
is circular (this is a consequence of using the Church encoding of words),
as we take $\symbf{c}_{n+1}=\symbf{c}_{0}=\star$.
Flows representing the word $\symbf{c}_1\cdots \symbf{c}_n$ can be pictured as follows:
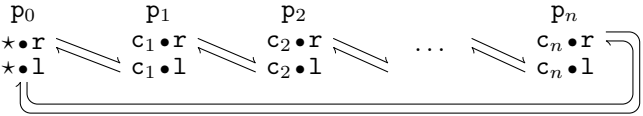
\label{dessin-inputs}
\begin{center}
\begin{tikzpicture}[x=0.9cm,y=1.2cm,inner sep=2pt, outer sep=1pt]
\node (S) at (0,0) {$\symbf{p}_{0}$};
\node (SS) at (S.south) {};
\node (So) at ($(SS.south)-(0,0.05)$) {$\star\p \rgt$};
\node (Si) at ($(So.south)-(0,0.1)$) {$\star\p \lft$};
\node (SiL) at (So.north) {};
\node (SoL) at (Si.north) {};
\node (A) at ($(S)+(2,0)$) {$\symbf{p}_{1}$};
\node (AS) at (A.south) {};
\node (Ao) at ($(AS.south)-(0,0.05)$) {$\symbf{c}_{1}\p \rgt$};
\node (Ai) at ($(Ao.south)-(0,0.1)$) {$\symbf{c}_{1}\p \lft$};
\node (AiL) at (Ao.north) {};
\node (AoL) at (Ai.north) {};
\node (B) at ($(A)+(2,0)$) {$\symbf{p}_{2}$};
\node (BS) at (B.south) {};
\node (Bo) at ($(BS.south)-(0,0.05)$) {$\symbf{c}_{2}\p \rgt$};
\node (Bi) at ($(Bo.south)-(0,0.1)$) {$\symbf{c}_{2}\p \lft$};
\node (BiL) at (Bo.north) {};
\node (BoL) at (Bi.north) {};
 \node (C) at ($(B)+(2,0)$) {};
 \node (points) at ($(B)+(2,-0.4)$) {$\dots$};
 \node (CS) at ($(BS)+(1.5,0)$) {};
 \node (Co) at ($(Bo)+(1.5,0)$) {};
 \node (Ci) at ($(Bi)+(1.5,0)$) {};
 \node (CiL) at ($(BiL)+(1.5,0)$) {};
 \node (CoL) at ($(BoL)+(1.5,0)$) {};
\node (D) at ($(C)+(2,0)$) {$\symbf{p}_{n}$};
\node (DS) at (D.south) {};
\node (Do) at ($(DS.south)-(0,0.05)$) {$\symbf{c}_{n}\p \rgt$};
\node (Di) at ($(Do.south)-(0,0.1)$) {$\symbf{c}_{n}\p \lft$};
\node (DiL) at (Do.north) {};
\node (DoL) at (Di.north) {};
\node (Do2) at ($(Di)+(1,0)$) {};
\draw[-left to] ($(So.east)+(0,0.05)$) to ($(Ai.west)+(0,0.05)$) {};
\draw[-left to] ($(Ai.west)+(0,-0.05)$) to ($(So.east)+(0,-0.05)$) {};
\draw[-left to] ($(Ao.east)+(0,0.05)$) to ($(Bi.west)+(0,0.05)$) {};
\draw[-left to] ($(Bi.west)+(0,-0.05)$) to ($(Ao.east)+(0,-0.05)$) {};
\draw[-left to] ($(Bo.east)+(0,0.05)$) to ($(Ci.west)+(0,0.05)$) {};
\draw[-left to] ($(Ci.west)+(0,-0.05)$) to ($(Bo.east)+(0,-0.05)$) {};
 \draw[-left to] ($(Co.east)+(1,0.05)$) to ($(Di.west)+(0,0.05)$) {};
 \draw[-left to] ($(Di.west)+(0,-0.05)$) to ($(Co.east)+(1,-0.05)$) {};
\draw[-left to, rounded corners] ($(Si.south)+(0.05,0)$)  -- (0.05,-1) -- (9,-1) -- (9, -0.3) -- (8.6, -0.3);
\draw[-left to, rounded corners] (8.6, -0.2) -- (9.1, -0.2) -- (9.1,-1.1)  -- (-0.05,-1.1) --  ($(Si.south)+(-0.05,0)$);
\end{tikzpicture}
\end{center}

On the other hand, the notion of \emph{observation} will be the counterpart of a program in our construction. We
first give a general definition, that will be instantiated later to classes of observations that characterize specific complexity classes.
The important point here
is that we forbid an observation to use any position constant, in order to have it interact the same way
with all the representations $\wordrep Wp$ of a word $W$.

\begin{definition}[observation semiring]
We define the semirings $\unused P$ of flows that do not use the symbols in $\symbf P$; and $\alphabet$
the semiring generated by flows of the form $\symbf c\p \symbf d\flow\symbf c'\p \symbf d'$
with $\symbf c,\symbf c'\in\Sigma\cup\{\dummy\}$ and $\symbf d,\symbf d'\in \LR$.

We define the semiring of \emph{observations} as:
\[\algfont O\eqdef(\alphabet\p\usring)\cap\unused P\]
and the semiring of \emph{observations over the semiring $\algfont A$} as%
\[\simpleoring{\algfont A}\eqdef(\alphabet\p\algfont A)\cap\unused P\]
\end{definition}

The following theorem is a consequence~\cite[Theorem~IV.5]{Bagnol2014} of the fact that
observations cannot use position constants.

\begin{theorem}[normativity]
\label{normativityppty}
Let $\bar{W}_{p}$ and $\bar{W}_{q}$ be two representations of a word $W$ and $O$ an
observation.

Then $\isnilp{\obs{O}\bar{W}_{p}}$ if and only if $\isnilp{\obs{O}\bar{W}_{q}}$.
\end{theorem}

With this theorem, we can safely define how a word can be accepted by an observation: %
the notion is independent of the specific choice of a representation of position constants.

\begin{definition}[accepted language]
\label{lang-obs}
Let $\obs{O}$ be an observation.
We define the \emph{language accepted by $\obs{O}$} as
\[\lang{\obs{O}}:=\set{W\in \Sigma^{\ast}}{\forall p, \,\isnilp{\obs{O}\bar{W}_{p}}}\]
\end{definition}

\subsection{Balanced Flows and Logarithmic Space}\label{sec_logspace}

In previous work~\cite{Aubert2014b}, we investigated the semiring of \emph{balanced wirings}, that are defined as sets of balanced \incise{or \enquote{height-preserving}} flows.

\begin{definition}[balance]%
\label{def-balanced}
A flow $f=\term{t}\flow \term{u}$ is \emph{balanced} if for any variable $x\in\Var(\term{t})\cup\Var(\term{u})$,
all occurrences of $x$ in both \(\term{t}\) and \(\term{u}\) %
have the same height (recall notations p.~\pageref{def-height}).
A \emph{balanced} wiring $F$ is a sum of balanced flows.

We write $\balanced$ for the set of balanced wirings.%
\end{definition}

\begin{definition}[balanced observation]
A \emph{balanced observation} is an element of
$\simpleoring{\balanced\p\balanced}$.
\end{definition}

This natural restriction was shown to characterize logarithmic space
computation~\cite[Theorems 34-35]{Aubert2014b}.

\begin{theorem}[balance and logarithmic space]
If $O$ is a balanced observation, then $\lang{\obs{O}}\in\coNLogspace$.
If $L\in\coNLogspace$ then there exists a balanced observation such that $\lang{\obs{O}}=L$.
\end{theorem}

It also appears that a natural subclass of balanced wirings characterizes \DLogspace, the class of
\emph{deterministic} logarithmic space computable predicates.

\subsection{The \texorpdfstring{$\stack$}{Stack} Semiring}\label{sec_stack}

This paper deals with another restriction on flows, namely the restriction to \emph{unary flows}, \ie flows defined from unary function symbols only. The semiring of wirings composed only of unary flows is called the $\stack$ semiring, and will be shown to characterize polynomial time computation. Here we briefly give the definitions and results about this semiring that will be needed in this paper. A more complete picture can be found in the second author's Ph.D. thesis~\cite%
{Bagnol2014}.%

\begin{definition}[unary flows]
A \emph{unary flow} is a flow built using only unary function symbols
and a variable.

The semiring $\stack$ is the set of wirings of the form \hbox{$\sum_i \term{t}_i\flow \term{u}_i$}
where the $\term{t}_i\flow \term{u}_i$ are unary flows.
\end{definition}

\begin{example}
The flows $\symbf{f(f(\var x))\flow g(\var x)}$ and $\symbf{\var x\flow g(\var x)}$ are unary,
while $x\p \symbf f(x)\flow \symbf g(x)$ and $\symbf{f(c)\flow \var x}$ are not.
\end{example}

\begin{notation}[stack operations]\label{stack-op-def}
If $\tau=\extlist{\symbf g_1}{\symbf g_n}$ is a finite sequence of unary function symbols and $\term{t}$
is a term, we write $\tau(\term{t})\eqdef\symbf g_1\big(\symbf g_2(\cdots\symbf g_n(\term{t})\cdots\big)$.
We will write $\tau\sigma$ the concatenation of the sequences $\tau$ and $\sigma$.
Given two sequences $\tau$ and $\sigma$ we define the flow
$\op\tau\sigma\eqdef\, \tau(x) \flow \sigma(x)$
which we call a \emph{stack operation}.
\end{notation}

Note that, by definition, an element of the $\stack$ semiring must be a sum of stack operations.

The notion of \emph{cyclic flow} is crucial in the proof of the characterization of polynomial
time computation.
As we will see, it is complementary to the nilpotency property for
elements of $\stack$, \ie a wiring in $\stack$ will be shown to be either cyclic or nilpotent.

\begin{definition}[cyclicity]
A flow $\term{t}\flow \term{u}$ is a \emph{cycle} if $\term{t}$ and $\term{u}$ are matchable (\autoref{disjoint}).
A wiring $F$ is \emph{cyclic} if there is a $k$ such that $F^k$ contains a cycle.%
\label{def-cyclicity}

For $\sequ s=\extlist{f_1}{f_n}$ a sequence of stack operations, define:
\begin{itemize}
\item its \emph{height}
as $\h(\sequ s)\eqdef \max_i \big\{\h(f_i)\big\}$
\item its \emph{cardinality}%
\footnote{Note that the cardinality of $\sequ s$ is not necessarily equal to the length of $\sequ s$. For instance, if $\sequ s= f_{1},f_{1},f_{2}$ with $f_{1}\neq f_{2}$ then $\card(\sequ s)=2$.}
$\card(\sequ s)\eqdef \card\set{f_i}{1\leq i \leq n}$. %
\item its product $\pr(\sequ s)$ as $f_1\cdots f_n$.
\end{itemize}

We say the sequence $\sequ s$ is \emph{cyclic} if there is a sub-sequence
$\sequ s_{i,j}=f_i,\dots, f_j$ (\hbox{$1\leq i\leq j\leq n$})
such that $\pr(\sequ s_{i,j})$ is a cycle.
\end{definition}

\begin{remark}\label{rem_cycl}
A flow $f$ is a cycle \iff $f^2\neq0$.
\end{remark}

To carry on the proof evoked above that cyclicity and nilpotency are complementary notions in $\stack$, we borrow a
result from an earlier work about \GoI and complexity in the context of an algebra of Horn clauses.

\begin{lemma}[acyclic sequence \protect{\cite[lemma~5.3]{Baillot2001}}]\label{lem_baillot}
If $\sequ s$ is an acyclic sequence of stack operations, then we have
\[\h\big(\pr(\sequ s)\big)\leq \h(\sequ s)(\card(\sequ s)+1)\]
\end{lemma}

The following property says that cycles in $\stack$ can be iterated indefinitely, \ie a stack operation
$\op\tau\sigma$ such that $(\op\tau\sigma)^2\neq0$ is never nilpotent.

\begin{proposition}\label{cor_cycl}
If a stack operation $f$ is a cycle, then $f^n\neq 0$ for all $n$.
\end{proposition}

\begin{remark}
This does not hold for general flows. For instance, $f=\symbf{\var x\p c \flow d \p \var x}$ is a cycle because $f^2=\symbf{c\p c\flow d\p d}\neq 0$ (by \autoref{rem_cycl}), but $f^3=(\symbf{\var x\p c \flow d \p \var x})(\symbf{c\p c\flow d\p d})=0$. %
\end{remark}

\begin{theorem}[nilpotency in $\stack$]
A wiring $F\in\stack$ is nilpotent \iff it is acyclic.
\end{theorem}

\begin{proof}
Suppose $F$ is not nilpotent, so that there is at least one stack operation in $F^n$ for any $n$, and let $S$ be the number of different function symbols appearing in $F$.
Set $k\eqdef (S^{\h(F)(\card(F)+1)}+S^{\h(F)(\card(F)+1)-1}+\cdots+1)^2$, \ie the total number of different flows of height at most $\h(F)(\card(F)+1)$ using the symbols appearing in~$F$.

Let $f\neq 0$ be an element of $F^{k+1}$. It is the product $\pr(\sequ s)$ of a sequence $\sequ s=\extlist{f_1}{f_{k+1}}$ of stack operations that belong to $F$. We show by contradiction that this sequence must be cyclic, so let us suppose it is not.
By \autoref{lem_baillot}, we know that for any $i>0$, setting $\sequ s_i\eqdef\extlist{f_1}{f_i}$ we have
\[\h\big(\pr(\sequ s_i)\big)\leq\h(\sequ s_i)(\card(\sequ s_i)+1)\leq\h(F)(\card(F)+1)\]
Therefore, for any $i>0$ the flow $\pr(\sequ s_i)$ is of height at most $\h(F)(\card(F)+1)$ and uses only symbols appearing in $F$, \ie it wanders in a set of cardinal $k$, so there must be $1\leq i<j\leq k+1$ such that $\pr(\sequ s_i)=\pr(\sequ s_j)$.

Now, setting $\sequ s_{i+1,j}\eqdef\extlist{f_{i+1}}{f_j}$, we have that $\pr(\sequ s_i)\pr(\sequ s_{i+1,j})=\pr(\sequ s_j)=\pr(\sequ s_i)$ hence $\pr(\sequ s_i)\pr(\sequ s_{i+1,j})^2=\pr(\sequ s_i)\neq 0$ and thus $\pr(\sequ s_{i+1,j})^2\neq 0$ \ie $\pr(\sequ s_{i+1,j})$ is a cycle.
As $\pr(\sequ s_{i+1,j})\in F^{j-i}$ we can conclude that $F$ is cyclic.

The converse is an immediate consequence of \autoref{cor_cycl}.
\end{proof}

\begin{example}\label{ex_count}
Consider the wiring
\[\begin{array}{lll}
F\eqdef &   & \symbf f_1(x) \flow \symbf f_0(x) \\
        & + & \symbf f_0(\symbf f_1(x)) \flow \symbf f_1(\symbf f_0(x)) \\
        & + & \symbf f_0(\symbf f_0(\symbf f_1(x))) \flow \symbf f_1(\symbf  f_1(\symbf f_0(x))) \\
        & + & \symbf f_0(\symbf f_0(\symbf f_0(x))) \flow \symbf f_1(\symbf f_1(\symbf f_1(x)))
\end{array}
\]
which implements a sort of counter from $0$ to $7$ in binary notation
that resets to $0$ when it reaches $8$
(we see the sequence $\symbf f_x\symbf f_y\symbf f_z$ as the integer $x+2y+4z$).
It is clear with this intuition in mind that this wiring is cyclic.
Indeed, an easy computation shows that
$\symbf f_0(\symbf f_0(\symbf f_0(x))) \flow \symbf f_0(\symbf f_0(\symbf f_0(x)))\in F^8$.

If we lift this example to the case of a counter from $0$ to $2^n-1$ that resets to $0$
when it reaches $2^n$, we obtain an example of a wiring $F$ of cardinal $n$
and height $n-1$ such that $F^{2^n}$ contains a cycle, but $F^{2^n-1}$ does not.
This shows that the number of iterations needed to find a cycle may be exponential in the height
and the cardinal of $F$, which rules out a polynomial time decision procedure for the nilpotency
problem that would simply compute the iterations of $F$ until it finds a cycle in it.
\end{example}

Finally, let us define a new class of observations, based on the $\stack$ semiring.

\begin{definition}%
\label{def_stobs}
A \emph{balanced observation with stack} is an element of
$\stackobs\eqdef\simpleoring{\stack\p\balanced}$.
\end{definition}

\sepsection
\section{Pushdown Automata and \Ptime Completeness}
\subsection{Characterization of \Ptime by Pushdown Automata}\label{sec_cook}
The class of deterministic polynomial time computable predicates \Ptime is the most studied complexity class,
mainly because it supposedly contains all \enquote{tractable} problems.

Extending our approach to this class was a long-standing goal,
whose completeness part is attained thanks to the connection with pushdown automata.
In this subsection, we recall their definition, the \Ptime characterization theorem we will rely on and the memoization technique.

\subsubsection{Definition and classical results}
\label{def-automata}
Automata form a very basic model of computation that can be extended in different ways.
For instance, allowing multiple heads that can move in two directions on the input tape, one gets a model of computation equivalent to read-only Turing machines.

Among possible extensions, our interest will focus on the addition of a \enquote{pushdown stack} (together
with multiple heads), which we referred to as \enquote{pushdown automata} until now. We will see that this
leads to a characterization of \Ptime.

Let us give below the most general definition, for the non-deterministic case.

\begin{definition}[\sansparamboldmnfas]%
\label{ndfa_def}
For \(k \geqslant 1\), a \emph{\(2\)-way \(k\)-head
finite automaton with pushdown stack (\mnfas{k})} is a tuple \(\automate{M} = \{ \state{S}, \state{i}, A, B, %
 \tridro, \trigau, \bstack, \sigma \}\) where:
\begin{itemize}
\item \(\state{S}\) is the finite set of \emph{states}, with \(\state{i} \in \state{S}\) the \emph{initial state}%
;
\item \(A\) is the \emph{input alphabet}, \(B\) the \emph{stack alphabet};
\item \(\tridro\) and \(\trigau\) are the \emph{left} and \emph{right endmarkers}, \(\tridro, \trigau \notin A\);
\item \(\bstack\) is the \emph{bottom symbol of the stack}, \(\bstack \notin B\);
\item \(\sigma\) is the \emph{transition relation}, \ie a subset of the product \((\state{S} \times (A_{\bowtie}\})^k \times B_{\bstack}) \times (\state{S} \times \{-1, 0, +1\}^k \times \{\pop, %
\push(b)\}) \) %
where \(A_{\bowtie}\) (\resp \(B_{\bstack}\)) denotes \(A \cup \{ \tridro, \trigau\}\) (\resp \(B \cup \{ \bstack\}\)).
The instruction \(-1\) corresponds to moving the head one cell to the left, \(0\) corresponds to keeping the head on the current cell and \(+1\) corresponds to moving it one cell to the right.
Regarding the pushdown stack, the instruction \(\pop\) means \enquote{erase the top symbol}, %
while, for all \(b \in B\), \(\push(b)\) means \enquote{write \(b\) on top of the stack}.
\end{itemize}
\end{definition}

The automaton \emph{rejects the input} if it loops, otherwise it \emph{accepts}.
This condition is equivalent to the standard way of defining acceptance and rejection by \enquote{reaching a special state}~\cite[Theorem~2, p.~125]{Ladermann1994}.
Modulo another standard transformation, we restrict the transition relation so that at most one head moves at each transition.

Without pushdown stacks, \(2\)-way \(k\)-head finite automata
characterize \(\Logspace\)
and \(\NLogspace\), depending on the automata being deterministic or not.

This result, used in our previous work~\cite{Aubert2014b,Aubert2015temp}, was first stated informally by Juris Hartmanis~\cite[pp.~338--339]{Hartmanis1972}
and is often \cite[p.~13]{Cook1971}, \cite[pp.~338--339]{Hartmanis1972}, attributed to Alan Cobham. However, a detailed proof can be found in a classical
handbook~\cite[pp.~223--225]{Wagner1986}.
The addition of a pushdown stack improves the expressivity of the machine model, as stated in the following theorem.

\begin{theorem}\label{th_autom-P}
\sansparammnfas characterize \Ptime.
\end{theorem}

Without reproving this classical result of complexity theory, we review the main ideas that support it.

\textbf{Simulating a polynomial-time Turing machine with a \sansparamboldmnfas} amounts to designing an equivalent Turing machine whose movements of heads follow a regular pattern.
That permits to seamlessly simulate their contents with a pushdown stack.
A complete proof~\cite[pp.~9--11]{Cook1971} as well as a precise algorithm~\cite[pp.~238--240]{Wagner1986}
can be found in the literature.

\textbf{Simulating a \sansparamboldmnfas with a polynomial-time Turing machine} cannot amount to simply simulate step-by-step the automaton with the Turing machine.
The reason is that for any automaton, one can design an automaton that recognizes the same language but runs exponentially slower \cite[p.~197]{Aho1968}.
That the automaton can accept its input after an exponential computation
time is similar with the situation of the counter in \autoref{ex_count}.

The technique invented by Alfred V.~Aho \etal~\cite{Aho1968} and made popular by Stephen A.~Cook consists in building a \enquote{memoization table} that allows the Turing machine to create shortcuts in the simulation of the automaton, decreasing drastically its computation time.
In some cases, an automaton with an exponentially long run can even be simulated in linear time~\cite{Cook1971a}.

We give more details on this technique in the next subsection, as its adaptation to our context will be a key
ingredient in the soundness proof in \autoref{sec-nilp-of-stack}.

\subsubsection{The memoization technique}
Although the name comes from machine-learning~\cite{Michie1968}, this technique is usually attributed to S.~A.~Cook and has provided fundamental as well as practical results.
In the specific case of automata with stack, it can be condensed in the following remark:
if at a given time you are in state $\state{q}$ with \(b\) on top of a stack of height \(h \geqslant 1\), and if you end up later on in the state $\state{q}'$ with some symbol \(b'\) on top of a stack of height \(h\), \emph{without having popped a symbol at height inferior to \(h\)}, and if you are about to pop this symbol, then you can save this progression
$(\state q,b)\rightarrow(\state q',b')$.
If later on you find yourself in the same state \state{q}, with \(b\) on top of your stack and with the heads in the same positions, you can directly skip to the saved progression, as there is no need to perform this part of the computation again.
This \enquote{partial information}, the description of your automaton without the contents of the stack, apart from its top symbol, is sometimes called \enquote{surface configuration} or \enquote{partial identifier}.

The memoization technique consists in building and using the transitive closure of the relation between surface configuration.
Differently expressed, memoization is a \enquote{\emph{clever evaluation strategy}, applicable whenever the results of certain computations are needed more than once}~\cite[p.~348]{Amtoft1992}.
One looking for subtle refinements could look for a technique of memoization computed independently from the input, allowing to \enquote{compile} a stack program into equivalent online memoizing program~\cite{Andersen1994} that runs exponentially faster.
A nice explanation in the case of single head automata can be found in a recent and short article by R. Glück~\cite{Gluck2013}.

We will be adapting this idea to our context in \autoref{saturation-section}, which will
amount to a form of \emph{exponentiation by squaring}.

\subsection{Encoding \texorpdfstring{\sansparammnfas}{2MFA+S} as Observations: \texorpdfstring{$\Ptime$}{Ptime} Completeness}\label{sec_completeness}
The encoding proposed below is similar to the previously developed~\cite[Sect.~4.1]{Aubert2014b} encoding of \(2\)-way \(k\)-head finite automata (without pushdown stack) by flows.
The only difference is the addition of a \enquote{plug-in} that allows for a representation of stacks in observations.

Remember that acceptance by observations is phrased in terms of nilpotency of the
product $O\wordrep W p$ of the observation and the representation of the input (\autoref{lang-obs}).
Hence the computation in this model is defined as an iteration: one computes by considering the sequence $O\wordrep W p, (O\wordrep W p)^{2}, (O\wordrep W p)^{3}, \dots$ and the computation either ends at some point (\ie accepts) \incise{that is $(O\wordrep W p)^{n}=0$ for some integer $n$} or loops (\ie rejects).
One can think of this iteration as representing
a dialogue, or a game, between the observation and its input.

We turn now to the proof of \emph{\Ptime-completeness} for the set of balanced observations with stacks.

\begin{theorem}\label{th_pcomplete}
If \({L} \in\Ptime\), then there exists a balanced observation with stack \(O\in\stackobs\)
such that \({L} = \lang O\).
\end{theorem}

\newtoggle{decription} %
\togglefalse{decription} %
\iftoggle{decription}{
\newcommand{\marqueur}[1]{\item[#1]}
}{
\newcommand{\marqueur}[1]{\noindent\textbf{#1}}
}
\iftoggle{decription}{\begin{manual_qed_proof}}{\begin{proof}}

The proof relies on encoding a \mnfas{k} \(\automate{M}\) that recognizes \(\mathcal{L}\) \incise{whose existence is ensured
by \autoref{th_autom-P}} as an observation of \(\stackobs\).
Taking \(A = \Sigma\) the input alphabet, \(k +1 \) the number of heads of the automaton,
we will encode the transition relation of $M$ as a
balanced observation with stack. More precisely, the automaton will be represented as an element $O_{\automate{M}}$ of
\(\stackobs=\simpleoring{\stack\p\balanced}\) which can be written as a sum of flows of the form
\[
\begin{multlined}
\symbf c'
\p
\symbf d'
\p
\sigma(x)
\p
\symbf q'
\p
\nary k (y_1', \hdots, y_k')
\p
\Mbis (z')
\flow \\
\symbf c
\p
\symbf d
\p
\symbf s (x)
\p
\symbf q
\p
\nary k (y_1, \hdots, y_k)
\p
\Mbis (z)
\end{multlined}
\]
with
\begin{itemize}
\item \(\symbf c,\symbf c'\in\Sigma\cup\{\dummy\}\),
\item \(\symbf d, \symbf d'\in \LR\),
\item \(\sigma\) a finite sequence of unary function symbols,
\item \(\symbf s\) a unary function symbol,
\item \(\symbf{q}, \symbf{q}'\) two constant symbols,
\item  \(\nary k, \Mbis\) functions symbols of respective arity \(k\) and \(1\).
\end{itemize}

The intuition behind the encoding is that a configuration of a %
\mnfas{k+1} processing an input can
be seen as a closed term
\[
\symbf c
\p
\symbf d
\p
\tau (\bstack)
\p
\symbf q
\p
\nary k (\extlist{\symbf p_{i_1}}{\symbf p_{i_k}})
\p
\Mbis (\symbf p_{j})
\]
where the $\symbf p_i$ are position constants representing the positions of the \emph{main pointer}
($\Mbis (\symbf p_{j})$) and of the \emph{auxiliary pointers}
($\nary k (\extlist{\symbf p_{i_1}}{\symbf p_{i_k}})$); the symbol $\symbf q$ represents the state the
automaton is in; $\tau(\bstack)$ represents the current stack; the symbol $\symbf d$ represents the direction
of the next move of the main pointer; the symbol $\symbf c$ represents the symbol currently read by the
main pointer.

When a configuration matches the right side of the flow, the transition is followed, leading to an
updated configuration.

More precisely, we will be encoding $M$ as an observation $O_{\automate{M}}$, and observe %
the iterations of $O_{\automate{M}}\wordrep Wp$, its product with a word representation.
Let us explain how the basic operations of $\automate{M}$ are encoded:

\iftoggle{decription}{\begin{description}}{}

\marqueur{Moving the pointers.}
Looking back at the definition of the encoding of words (\autoref{word-rep}), we see that we can have a
new reading of what the representation of a word does: it moves the main pointer in the required direction.
From that perspective, the position holding the symbol $\star$ in \autoref{word-rep} allows to
simulate the behavior of the endmarkers \(\tridro\) and \(\trigau\).

On the other hand, the observation is not able to manipulate the position of pointers directly (remember
observations are forbidden to use the position constants) but can change the direction symbol $\symbf d$,
rearrange pointers (hence changing which one is the main pointer) and modify its state and the symbol $\symbf c$
accordingly.
For instance, a flow of the form
\[
\begin{multlined}
\cdots \p \nary k (x, \hdots, y_k) \p \Mbis (y_1)
\flow \\
\cdots \p \nary k (y_1, \hdots, y_k) \p \Mbis (x)
\end{multlined}
\]
encodes the instruction \enquote{swap the main pointer and the first auxiliary pointer}.

Note however that our model has no built-in way to remember the values of the auxiliary pointers
\incise{it remembers only their \emph{positions} as arguments of $\nary k(\cdots)$},
but this can be implemented easily using additional states.

One can see that it is the interaction between the observation $O_{\automate{M}}$
and the word representation $\wordrep W p$ that
 simulates the behavior of the automaton, and not the observation on its own manipulating some passive data.

\marqueur{Handling the stack.}
Suppose we have a unary function symbol \(\func b(\cdot)\) for each symbol \(b\) of the stack alphabet
\(B_{\bstack}\).

A transition that reads \(b\) and pops it is simply written as
\[\cdots\p x\p\cdots \flow \cdots\p \func b(x)\p\cdots\]

A transition that reads \(b\) and pushes a symbol \(c\) is written
\[\cdots\p\func c(\func b(x))\p\cdots \flow \cdots\p \func b(x) \p\cdots\]

\marqueur{Changing the state.}
We suppose that we have a constant $\symbf q$ for each state $\state q$ of $\automate{M}$. Then, updating the state
amounts to picking the right $\symbf q$ and $\symbf q'$ in the flow representing the transition.

\marqueur{Acceptance and rejection.}
The encoding of acceptance and rejection is slightly more delicate, as detailed
in a previous article~\cite[6.2.3.]{Aubert2014ctemp}.

The basic idea is that acceptance in our model is defined as nilpotency, that is to say: the absence of
loops.
If no transition in the automaton can be fired, then no flow in our encoding can be unified,
and the computation ends.

Conversely, a loop in the automaton will refrain the wiring from being nilpotent.
The point we need to be careful about is the encoding of loops: those should be represented as a re-initialization of the computation, as discussed in details in earlier work \cite{Aubert2014ctemp}.
The reason for this is that another encoding may interfere with the representation of acceptation as termination: the existence of a loop in the observation $O_{\automate{M}}$ representing the automaton $\automate{M}$, \emph{even one that is not used in the computation with the input $W$}, prevents the wiring $O_{\automate{M}}\wordrep W p$ from being nilpotent.

Indeed, the \enquote{loop} in \autoref{ndfa_def} of \sansparammnfas is to be read as
\enquote{perform forever the same computation}.
\iftoggle{decription}{\mqed \end{description}\end{manual_qed_proof}}{\end{proof}}

Notice that the encoding of pushdown automata as observations with stacks produces only specific observations, namely those that are sums of flows of a particular form (shown at the beginning of the preceding proof). This is due to the fact that one encodes the transitions directly, so that each flow corresponds to a transition step.

In particular, %
as the transition relation of automata depends only on the top of the stack, the body (\ie the right-hand part) of the flows must be of the form $\cdots\p \func b(x)\p\cdots$.
However, a general observation with stack is not constrained in this way, and allows a more compact representation of programs where one can read, pop and push several symbols of the stack simultaneously.

Nevertheless, this does not increase the expressive power: the next section is devoted to prove that the language recognized by any observation with stack lies in \Ptime.

\sepsection
\section{Nilpotency in \texorpdfstring{$\stack$}{Stack} and \Ptime soundness}\label{sec-nilp-of-stack}
\subsection{The Saturation Technique}\label{saturation-section}

We now introduce the \emph{saturation} technique, which allows to decide nilpotency of $\stack$ elements in
polynomial time.
This technique relies on the fact that under certain conditions, the height of flows does not grow when computing
their product.
It adapts memoization to our setting: we repeatedly extend the wiring by adding pairwise products of flows,
allowing for more and more \enquote{transitions} at each step.

\begin{notation}\label{not_incr}
Let $\tau$ and $\sigma$ be sequences of unary function symbols.

If $\h\big(\tau(x)\big)\geq\h\big(\sigma(x)\big)$ we say that $\op\tau\sigma$ is \emph{increasing}.

If $\h\big(\tau(x)\big)\leq\h\big(\sigma(x)\big)$ we say that $\op\tau\sigma$ is \emph{decreasing}.

A wiring in \(\stack\) is \emph{increasing} (\resp \emph{decreasing}) if it contains only increasing
(\resp \emph{decreasing}) stack operations.
\end{notation}

\begin{lemma}[stability of height]
\label{lem_grow}
Let $f=\op\tau\sigma$ and $g=\op\rho\chi$ be stack operations.
If $f$ is decreasing and $g$ is increasing, we have $\h(fg)\leq \max\{\h(f),\h(g)\}$.
\end{lemma}

\begin{proof}
If $fg=0$, the property holds because $\h(0)=0$.
Otherwise, we have either $\sigma=\rho\mu$ or $\sigma\mu=\rho$.

Suppose we are in the first case (the second being symmetric).
Then we have $fg=\op{\tau}{\chi\mu}$ and $\h(\sigma)=\h(\rho\mu)$.

As $g$ is increasing, $\h(\chi)\leq\h(\rho)$ and therefore we have $\h(\chi\mu)\leq\h(\rho\mu)=\h(\sigma)\leq \h(f)\leq\max\{\h(f),\h(g)\} $.
\end{proof}

With this lemma in mind, we can define a \emph{shortcut} operation that augments an element of $\stack$ by
adding new flows while keeping the maximal height unchanged. Iterating this operation, we obtain a \emph{saturated} version
of the initial wiring, containing shortcuts, shortcuts of shortcuts, \etc

We are designing in fact an \emph{exponentiation by squaring} procedure for elements of $\stack$, the algebraic reading of memoization.

\begin{definition}[saturation]\label{def_sat}
If $F\in\stack$ we define its increasing $\inc F\eqdef\set{f\in F}{f \text{ is increasing}}$
and decreasing $\dec F\eqdef\set{f\in F}{f \text{ is decreasing}}$ subsets.%

We set the \emph{shortcut} operation $\short(F)\eqdef F+\dec F\inc F$ and its least fixpoint, which we call the \emph{saturation} of $F$:
\[\sat(F)\eqdef \sum_{n\in\NN}\short^n(F)\]
where $\short^{n}$ denotes the $n$\textsuperscript{th} iteration of $\short$.
\end{definition}

Now, as we are only manipulating flows with a limited height, the iteration of the shortcut operation is bound to stabilize at some point.

\begin{proposition}[stability of saturation]\label{prop_sat}
Let $F\in\stack$ be a wiring and $S$ the number of distinct function symbols appearing in $F$.

For any $n$, we have $\h\big(\short^n(F)\big)=\h(F)$.

Moreover if $n\geq (S^{\h(F)}+S^{\h(F)-1}+\cdots+1)^2$ then $\short^{\,n}(F)=\sat(F)$.
\end{proposition}

\begin{proof}
By \autoref{lem_grow} we have
\[\h(\dec F\inc F)\leq\max\{\h(\dec F),\h(\inc F)\}=\h(F)\]
Therefore $\h\big(\short(F)\big)= \h(F)$, and we get the first property by induction.

For any $n$, the elements of $\short^n(F)$ are stack operations of height at most $\h(F)$ built
from the function symbols appearing in $F$, therefore $\short^n(F)$ is a subset of a set of cardinality
$k\eqdef(S^{\h(F)}+S^{\h(F)-1}+\cdots+1)^2$. As $G\subseteq\short(G)$ for all $G$,
the iteration of $\short(\cdot)$ on $F$ must be stable after at most $k$ steps.
\end{proof}

In the following, we let \FPtime be the class of functions computable by Turing
machine in polynomial time. Here we need to specify how the size of a wiring is measured.

\begin{definition}[size]\label{def_size}
The \emph{size} $\size F$ of a wiring $F$ is defined as the total number of function symbol occurrences in
it.
\end{definition}

By computing the fixpoint of $\short(\cdot)$ we have first a \FPtime procedure computing the saturation.

\begin{corollary}[computing the saturation]\label{cor_sat} %
Given any integer $h$, there is procedure $\algo{Satur}{h}(\cdot)\in\FPtime$ that, given an element
$F\in\stack$ such that $\h(F)\leq h$ as an input, outputs $\sat(F)$.
\end{corollary}

Moreover, we can obtain a further reduction of the nilpotency problem in $\stack$ related to saturation.

\begin{lemma}[rotation]\label{lem_rota}
Let $f$ and $g$ be stack operations. Then $fg$ is a cycle \iff $gf$ is a cycle.
\end{lemma}

\begin{proof}
If $fg$ is a cycle, then $(fg)^n\neq 0$ for any $n$ by \autoref{cor_cycl}. In particular $(fg)^3\neq 0$
and as we have $(fg)^3=f(gf)(gf)g$ we get $(gf)^2\neq 0$, \ie $gf$ is a cycle.
\end{proof}

\begin{theorem}[cyclicity and saturation]\label{th_sat} %
An element $F$ of $\stack$ is cyclic (\autoref{def-cyclicity}) \iff either $\inc{\sat(F)}$ or $\dec{\sat(F)}$ is.
\end{theorem}

\begin{proof}
The cyclicity of $\inc{\sat(F)}$ or $\dec{\sat(F)}$ obviously implies that of $F$ because $\short(F)\subseteq F+F^2$, hence $\sat(F)\subseteq\sum_{n\in\NN}F^n$.

Conversely, suppose $F$ is cyclic and let $\sequ s=\extlist{f_1}{f_n}\in F$ be such that the product $\pr(\sequ s)\in F^n$ is a cycle.

We are going to produce from $\sequ s$ a sequence of elements of $\inc{\sat(F)}$ or $\dec{\sat(F)}$ whose product is a cycle. For this we apply to the sequence the following rewriting procedure:

\begin{enumerate}
\item \label{rew_1}If there are $f_i$ and $f_{i+1}$ such that $f_i$
is decreasing and $f_{i+1}$ is increasing, then rewrite $\sequ s$
as $\extlist{f_1,\dots,\,f_if_{i+1}\,}{f_n}$.
\item \label{rew_2}If step \ref{rew_1} does not apply and $\sequ s=\sequ s_1\sequ s_2$
($\sequ s_1$ and $\sequ s_2$ both non-empty) with all elements of $\sequ s_1$ increasing
and all elements of $\sequ s_2$ decreasing, then rewrite $\sequ s$ as $\sequ s_2\sequ s_1$.
\end{enumerate}

This rewriting procedure preserves the following invariants:
\begin{itemize}
\item All elements of the sequence are in $\sat(F)$: step \ref{rew_2} does not affect the elements of the sequence (only their order) and step \ref{rew_1} replaces the flows $f_i\in\dec{\sat(F)}$ and $f_{i+1}\in\inc{\sat(F)}$ by $f_if_{i+1}\in\sat(F)$.
\item The product $\pr(\sequ s)$ of the sequence is a cycle: step \ref{rew_1} does not alter $\pr(\sequ s)$ and step \ref{rew_2} does not alter the fact that $\pr(\sequ s)$ is a cycle by \autoref{lem_rota}.
\end{itemize}

The rewriting terminates as step \ref{rew_1} strictly reduces the length of the sequence and step \ref{rew_2} can never be applied twice in a row (it can be applied only when step \ref{rew_1} is impossible and its application makes step \ref{rew_1} possible).
Let $\extlist{g_1}{g_n}$ be the resulting sequence, as it cannot be reduced, the $g_i$ must be either all increasing or all decreasing.

Therefore, by the invariants above $\extlist{g_1}{g_n}$ is either a sequence of elements of $\dec{\sat(F)}$ or $\inc{\sat(F)}$ such that the product $g_1\cdots g_n$ is a cycle.
\end{proof}

Finally, we need a way to decide cyclicity of elements of $\stack$ that are
either increasing or decreasing.

\begin{lemma}\label{lem-Incr_h}
Given any integer $h$, there is a procedure $\algo{Incr}{h}(\cdot)\in\Ptime$ that, given an element
$F\in\stack$ which is either increasing or decreasing and
satisfying $\h(F)\leq h$ as an input, accepts \iff $F$ is nilpotent.
\end{lemma}

\begin{proof}
Let $S$ be a set of function symbols and $h$ an integer. We define the \emph{truncation wiring} associated to $S$ and $h$
\[\trunc {h,S}\eqdef\sum_{\mathclap{\tau=\extlist{\term f_1}{\term f_{h}}\in S}} \tau(\dummy) \flow\tau(x)\]
and set for the rest of the proof $T\eqdef \trunc {\h(F),E}$ where $E$ is the set of function symbols
occurring in $F$.

As it contains only flows of the form $\tau(\dummy)\flow \sigma(x)$, \ie with only one variable,
$T F$ is balanced and can be computed in polynomial time since $T$ is of polynomial size in $\size F$.

If $F$ is increasing, an easy computation shows that we have $(T F)^n=T F^n$.
From this, we deduce that $F$ is nilpotent \iff $T F$ is.
If $F=\sum_i \sigma_i(x) \flow \tau_i(x)$ is decreasing, we can
consider $F^\dagger\eqdef\sum_i \tau_i(x) \flow \sigma_i(x)$ which is increasing and nilpotent \iff $F$ is.

Then, as we know~\cite[p.~54]{Aubert2014b}, \cite[Theorem~IV.12]{Bagnol2014}
 the nilpotency problem for balanced wirings to be in \coNLogspace$\subseteq$\Ptime, we are done.
\end{proof}

\begin{theorem}[nilpotency is in \Ptime]\label{th-completness}
Given any integer $h$, there is a procedure $\algo{Nilp}{h}(\cdot)\in\Ptime$ that, given a $F\in\stack$
such that $\h(F)\leq h$ as an input, accepts \iff $F$ is nilpotent.
\end{theorem}

\begin{proof}
Simply take \(\algo{Nilp}{h}(\cdot) = \algo{Incr}{h} (\algo{Satur}{h}(\cdot))\).
By compositionality of \(\Ptime\) and \(\FPtime\) algorithms, this procedure is in \(\Ptime\).
\end{proof}

\begin{remark}\label{rem_flat}
All the results we gave in this section are parametrized by a height limit $h$, but this is only to ease
the
presentation. Indeed, it is possible to transform any element of $\stack$ with an unspecified height
into another element of comparable size but of height at most $2$, preserving its nilpotency.

More precisely: consider a flow $l=\sigma(x)\flow\tau(x)$, with $\sigma=\extlist{\symbf f_1}{\symbf f_m}$ and
$\tau=\extlist{\symbf g_1}{\symbf g_n}$.
Let us introduce new function symbols $\flatpop{i} (\cdot)$ and $\flatpush{j} (\cdot)$ for $1\leq i \leq m$ and
$1\leq j \leq n$. We can rewrite $l$ as the sum
\[
\begin{multlined}
\shoveleft{\symbf g_1(x)\flow \flatpush 1(x)\,+\,\flatpush 1(\symbf g_2(x))\flow \flatpush 2(x)\,+\,\cdots}\\[.3em]
\shoveleft[.2\columnwidth]{\cdots+\flatpush n(x)\flow\flatpop 1(x)+\cdots}\\[.3em]
\shoveleft{\cdots\,+\,\flatpop {m-1}(x)\flow \flatpop m(\symbf f_{m-1}(x))\,+\,\flatpop m(x) \flow \symbf f_m(x)}
\end{multlined}
\]
with the idea that instead of popping and pushing several symbols at the same time, we do this step by step:
we push/pop only one symbol and then leave a marker (either $\flatpop{i} (\cdot)$ or $\flatpush{j} (\cdot)$)
for the next operation to be performed.
Moreover, we see that this flow of size (\autoref{def_size})
 $\size l=m+n$ is transformed in a wiring containing three flows of size $2$
(the central and two extremal ones) and $(m-1)+(n-1)$ flows of size $3$; hence the sum has size $3\size l$.

When dealing with a wiring $W$, we can do the same by considering one family of
$\flatpop{i} (\cdot)$ and $\flatpush{j} (\cdot)$ symbols for each flow $l$ of $W$.
It is not hard to see that the resulting wiring $W_{\textnormal{flat}}$ has the same behavior as the original one in terms of nilpotency.
It is also clear that the height of $W_{\textnormal{flat}}$ is indeed $2$.
Finally, if we started with a wiring $W$ of size $N$, which is the sum of the sizes of the flows in it,
we get in the end that $\size {W_{\textnormal{flat}}}=3\size W$ using the one-flow case above.

This suggests by the way that the bound in \autoref{prop_sat} is probably too rough,
but a way to sharpen it still needs to be found.
\end{remark}

\subsection{\Ptime Soundness}
We will now use the saturation technique to prove that the language recognized by an observation with stack
belongs to the class \Ptime.
The important point in the proof is that, given an observation $O$ and a representation \(\rep{W}_{p}\) of a word \(W\), one can produce in polynomial time
 an element of $\stack$ whose nilpotency is equivalent to the nilpotency of $O\rep{W}_{p}$.
One can then decide the nilpotency of this element thanks to the procedure described
in the previous section.

\begin{proposition}\label{prop_red}
Let $O \in \stackobs$ be an observation with stack.
There is a procedure $\algo{Red}O(\cdot)\in\FPtime$ that, given a word $W$ as an input, outputs a wiring
$F\in\stack$ with $\h(F)\leq\h(O)$ such that $F$ is nilpotent \iff $O\wordrep W p$ is for any
choice of $\positions p$.
\end{proposition}

\begin{proof}[sketch~\protect{\cite[proposition IV.21]{Bagnol2014}}]
The idea is that the product $O\wordrep W p$ can be seen as an element of $\balanced\p\stack$.
Then, its balanced part can be replaced in polynomial time by closed terms without altering the nilpotency
in a way similar to what is done to treat the nilpotency of elements of $\balanced$~\cite{Aubert2014b}.

We are left with a flow $\sum_i \term t_i\p \sigma_i(x) \flow \term u_i\p \tau_i(x)$ such that \( \term t_i \flow \term u_i\) is balanced and \(\sigma_i(x) \flow \tau_i (x) \) is a stack operation,
and we can associate to each closed $\term t_i,\term u_i$,
 unary function symbols \(\func t_i(\cdot)\), \(\func u_i(\cdot)\),
and rewrite our flow as
$\sum_i \functerm t_i(\sigma_i(x)) \flow \functerm u_i(\tau_i(x))\in\stack$.
\end{proof}

\begin{theorem}[soundness]\label{th_psound}
If $O\in\stackobs$ is an observation with stack, then $\lang O\in\Ptime$.
\end{theorem}

\begin{proof}
We have, using the compositionality of \Ptime and \FPtime again,
that $\algo{Nilp}{\h(O)}(\algo{Red}O(\cdot))$ is a decision procedure in \(\Ptime\) for $\lang O$.
\end{proof}

\sepsection
\section{Unary Logic Programming}
In previous sections, we showed how
the $\stack$ semiring captures polynomial time computation.
As we already mentioned, the elements of this semiring correspond to a specific class of logic programs.

We cover in here the consequences in terms of logic programming of the results and techniques introduced so far.
The basic definitions and a list of previously known results \incise{that highlight the novelty of our result} regarding logic programming can be found in an extensive survey~\cite{Dantsin2001}.

As an illustration, we show in \autoref{sec_cvp} how the classical boolean circuit value problem (CVP)~\cite{Ladner1975} can be encoded as a unary logic program, thus providing an alternative proof of its inclusion in \(\Ptime\).

\subsection{Unary Queries}\label{unaryqueries}

\begin{definition}[data, goal, query]
A \emph{unary query} is a triple $\queryf Q=(D,P,G)$, where:
\begin{itemize}
\item $D$ is a set of closed unary terms (a \emph{unary data}),
\item $P$ is a an element of $\stack$ (a \emph{unary program}),
\item $G$ is a closed unary term (a \emph{unary goal}).
\end{itemize}

We say that the query $\queryf Q$ \emph{succeeds} if $G\dashv$ can be derived combining $d\dashv$ for $d\in D$ and the elements of $P$ by the resolution rule exposed in \autoref{semiring}, otherwise we say the query \emph{fails}.

The \emph{size} $\size{\queryf Q}$ of the query is defined as the total number of occurrences of symbols in it.
\end{definition}

To apply the saturation technique directly, we need to represent all the elements of the unary query (data, program, goal) as elements of $\stack$. This requires an encoding.

\begin{definition}[encoding unary queries]
We suppose that for any constant symbol $\symbf c$, we have a unary function symbol $\func c(\cdot)$.
We also need two unary functions, \(\func{START}(\cdot)\) and \(\func{ACCEPT}(\cdot)\).

To any unary data $D$ we associate an element of $\stack$:
\[\denc D\eqdef \set{\tau(\func c(x))\flow \symbf{START}(x)}{\tau(\symbf c)\in D}\]
and to any unary goal $G=\tau(\symbf c)$ we associate
\[\genc G\eqdef \symbf{ACCEPT}(x)\flow\tau(\func c (x))\]
\end{definition}

\begin{remark}
The program part $P$ of the query needs not to be encoded as it is already an element of $\stack$.
\end{remark}

Once a query is encoded, we can tell if it is successful or not using the language of the resolution semiring.

\begin{lemma}[success]
A unary query $\queryf Q=(D,P,G)$ succeeds if and only if
\[\symbf{ACCEPT}(x)\flow\symbf{START}(x)\in \genc GP^n\denc D \quad\text{for some }n\]
\end{lemma}

Then, we can show that the saturation technique applies to the problem of deciding whether a unary query accepts.
The proof uses the saturation technique (\autoref{saturation-section}) to rewrite a sequence of flows, adding to them \enquote{pre-computed} rewriting rules.

\begin{lemma}[saturation of unary queries]
A unary query $\queryf Q=(D,P,G)$ succeeds if and only if
\[\symbf{ACCEPT}(x)\flow\symbf{START}(x)\in\sat\big(\denc D+P+\genc G\big)\]
\end{lemma}

\begin{theorem}[\Ptime-completeness]
The \problem{UQuery} problem (given a unary query, is it successful?) is \Ptime-complete.
\end{theorem}

\begin{proof}
The lemma above, combined with \autoref{cor_sat}, ensures that the problem lies indeed in the class
\Ptime, \modulo the considerations on the height of \autoref{rem_flat}.

The hardness part follows from a variation on the encoding presented in \autoref{sec_completeness}
and the reduction derived from \autoref{prop_red}.
\end{proof}

\begin{remark}
We presented the result in a restricted form, to stay in line with the previous sections. However, it should be clear to the reader that
this construction would not be impacted if we allowed
\begin{itemize}
\item non-closed goals and data;
\item that in $\term{t} \flow \term{u}$ the variables of $\term{t}$ does not appear in $\term{u}$;
\item constants in the program part of the query.
\end{itemize}

A harder question is whether everything scales up to logic programs %
of the form $H\dashv \extlist{B_1}{B_n}$, with more than one formula on the right of $\dashv$.
Indeed we would no longer have obvious notions of \emph{increasing} or \emph{decreasing} (\autoref{not_incr})
clause anymore, and these are crucial to the saturation technique.
It is already known~\cite[pp.~386--387]{Dantsin2001} that in the case of \emph{propositional}
(\ie with no variables) logic programming, allowing more than one $B_i$ makes the combined complexity
(see \autoref{combined-comp} below)
switch from \Logspace to \Ptime: one can expect by analogy a higher complexity
than \Ptime in our unary case, but nothing has been proven yet.

\end{remark}

\begin{remark}\label{combined-comp}
In terms of complexity of logic programs, we are considering the \emph{combined complexity}~\cite[p.~380]{Dantsin2001}: every part of the query $\queryf Q=(D,P,G)$ is variable.
If for instance we fixed $P$ and $G$ (thus considering \emph{data complexity}), we would have a problem that is still in \Ptime, but it is unclear to us if it would be complete.
Indeed, the encoding of \autoref{sec_completeness} relies on a representation of inputs as plain programs,
and on the fact that the evaluation process is a matter of interaction between
programs rather than mere data processing.

\end{remark}

\subsection{Circuit Value Problem}\label{sec_cvp}

To illustrate our point in the introduction about rephrasing a problem with unary symbols to tell
whether it lies in \Ptime, we present an encoding of the classical \Ptime-complete
\emph{circuit value problem} (CVP)~\cite{Dantsin2001} as a unary query.

An instance of CVP is a boolean circuit composed of \andg, \org, \notg, \zerog and \oneg gates
and is \emph{accepted} if the circuit computes the value \oneg at its output gate.

More formally, we can see an instance of CVP as $(G,o)$ with $G$ an acyclic directed hypergraph%
\footnote{A \emph{directed hypergraph} is given as a set of vertices $V$ and a set of edges
$E\subseteq \powerset (V)\times\powerset(V)$. We say that $(S,T)\in E$ is an edge \emph{from} $S$ \emph{to}
$T$.

We consider labeled edges and write %
{\scriptsize$\gate{\extlist{x_1}{x_n}}{\mathtt k}{\extlist{y_1}{y_m}}$} an edge labeled by $\mathtt k$ from
$\extset{x_1}{x_n}$ to $\extset{y_1}{y_m}$.}
with a distinguished output vertex $o$ built with edges among
\[\ande abc \qquad \ore abc \qquad \note ab \qquad \zeroe {a} \qquad \onee {a}\]
such that any vertex is the target of exactly one edge.

First, we associate to each vertex $v$ of the graph a pair $\funv v(\cdot),\funvn v(\cdot)$,
of unary function symbols. Then to each edge $e$ we associate a flow $\enc e$ as follows:
\begin{align*}
\enc{\ande abc} \,\eqdef\, & \;\funv a(\funv b(x))\flow \funv c(x)\\
 & +\; \funvn a(x)\flow \funvn c(x)\;+\;\funvn b(x)\flow \funvn c(x)\\[4pt]
\enc{\ore abc} \,\eqdef\, & \; \funv a(x)\flow \funv c(x)\;+\;\funv b(x)\flow \funv c(x)\\
 & +\;\funvn a(\funvn b(x))\flow \funvn c(x)\\[4pt]
\enc{\note ab} \,\eqdef\, & \;\funvn a(x)\flow \funv b(x)\;+\;\funv a(x)\flow \funvn b(x)\\[4pt]
\enc{\zeroe {a}} \,\eqdef\, & \;x\flow \funvn a(x)\\[4pt]
\enc{\onee {a}} \,\eqdef\, & \;x\flow \funv a(x)
\end{align*}

The intuition behind this encoding is that we are handling a stack of needed values, $\funv v(\cdot)$
(\resp $\funvn v(\cdot)$) meaning \enquote{we need the value \oneg (\resp \zerog) at $v$}. The flows
associated to gates are then meant to handle this stack, popping and pushing needed values.

Then, to a circuit $(G,o)$ we associate the unary query
\[\big(\:\funv o(\dummy)\:,\ \sum_{\mathclap{e \text{ vertex of }G}} \, [e]\:,\:\dummy\:\big)\]

This query succeeds \iff the circuit computes the value \oneg at the gate $o$: the data
$\funv o(\dummy)$ initiates a stack with the intuitive meaning \enquote{we need the value \oneg at $o$},
the encodings of edges propagate the needed values to the point where they can be \enquote{popped} if the correct
{\footnotesize$\zeroe x$} or {\footnotesize$\onee x$} is available.
The query succeeds if we
can derive the goal $\dummy$ \incise{\ie the empty stack} with the intuitive
meaning \enquote{all the needed values have been provided}.

Note the parallel nature of this way of solving the problem: when we compute the saturation of (the encoding
of) the query, we unify the terms that match at any point of the circuit without having to worry in which
order we perform the operations.

For instance, the two elements $\funv a(x) \flow \funv o(\star)$ and $\funv b(x) \flow \funv o(\star)$ of
$\enc{\ore ab o(\dummy)}$ would be unified with \(o(\dummy) \flow \dummy\),
providing two flows \( \funv a(x) \flow \funv \dummy\) and \(\funv b(x)\flow \dummy\).
Those flows would be, at the next execution step, tested for unification against all (provided we respect
the increasing/decreasing discipline at work in \autoref{def_sat})
the other flows and so on, without having to wait to know whether $a$ or $b$ will hold the value $\oneg$.
A partial evaluation happens at any point of the graph, independently of the input:
\(\enc{\ande abc}\) and \(\enc{\note c d}\) will give after one step of evaluation the flows
\(\funv a (\funv b (x)) \flow \funvn d (x)\), \(\funvn a (x) \flow \funv d (x)\)
and \(\funvn b (x) \flow \funvn d(x)\).
The execution does not have to sequentially wait for the propagation of the needed values.

Finally, let us say a word about the stabilization time of $\sat(\cdot)$ (\autoref{def_sat}) in this case. Given a circuit with $S$ gates, we are dealing
with flows of height at most $2$, written with at most $S$ different symbols. In view of \autoref{prop_sat}
we have that the iterations of $\short(\cdot)$ will stabilize in at most $(S^2+S+1)^2$ steps.
A bound that is rough, due the absence of optimization and fine-grained analysis of the procedure.

\sepsection
\section{Perspectives}
This article extends modularly on our previous approaches %
\cite{Aubert2014ctemp, Aubert2015temp, Aubert2014, Aubert2014b}
to obtain a characterization of \Ptime, by adding a sort of \enquote{stack plugin} to observations.
This enhancement was guided by the intuition of a stack added to an automaton, allowing to move
from \Logspace to \Ptime and providing a decisive proof technique: memoization.

We saw that to a \emph{qualitative} constraint on the way memory is handled by automata corresponds a \emph{syntactical} restriction on flows.
These flows are evaluated in a setting inspired by the representation of inputs
in the interactive approach to the Curry-Howard correspondence \incise{\GofI},
which makes the complexity parametric in the program \emph{and} the input.
However, despite the evaluation being highly parallel and different from the step-by-step evaluation
performed by automata,
a precise simulation of pushdown automata by unary logic program is given, leading to complexity results.

We were able to adapt the mechanism of pre-computation of transitions, known as memoization,
in a setting where logic programs are represented as algebraic objects. %
This technique \incise{that we called the saturation technique} computes shortcuts in a logic program in order to decide its nilpotency faster.

This approach to complexity was earlier based on von Neumann algebras~\cite{Girard2012, Aubert2014ctemp,Aubert2015temp}
and now explore unification theory~\cite{Bagnol2014,Aubert2014,Aubert2014b}:
it is emerging as a meeting point for computer science, logic and mathematics. %
This raises multiple questions and perspectives.

A number of interrogations come from the relations of this work to proof theory.
First, we could consider the Church encoding of other data types \incise{trees for instance} and
define \enquote{orthogonally} set of programs interacting with them, wondering what is their computational nature.
In the distance, one may hope for a connection between our approach and ongoing work on higher order
trees and model checking;
all alike, one could study the interaction between observations and one-way integers \incise{briefly discussed in earlier work \cite{Aubert2014b}} or non-deterministic data.
Second, a still unanswered question of interest is to give an account of observations in terms of a proof-system.

One could also investigate possible relations with other models of computation, such as the interaction abstract machine~\cite{Danos1999} that already developed and used \incise{although with a different, much more logical, meaning} the notion of shortcut in the evaluation.

Finally, we also aim at representing functional computation, by considering a more general notion of observation that would allow for expressing the notion of output.

\sepsection
\bibliographystyle{IEEEtran}
\bibliography{standalone}

\begin{thebibliography}{10}
\providecommand{\url}[1]{#1}
\csname url@samestyle\endcsname
\providecommand{\newblock}{\relax}
\providecommand{\bibinfo}[2]{#2}
\providecommand{\BIBentrySTDinterwordspacing}{\spaceskip=0pt\relax}
\providecommand{\BIBentryALTinterwordstretchfactor}{4}
\providecommand{\BIBentryALTinterwordspacing}{\spaceskip=\fontdimen2\font plus
\BIBentryALTinterwordstretchfactor\fontdimen3\font minus
  \fontdimen4\font\relax}
\providecommand{\BIBforeignlanguage}[2]{{%
\expandafter\ifx\csname l@#1\endcsname\relax
\typeout{** WARNING: IEEEtran.bst: No hyphenation pattern has been}%
\typeout{** loaded for the language `#1'. Using the pattern for}%
\typeout{** the default language instead.}%
\else
\language=\csname l@#1\endcsname
\fi
#2}}
\providecommand{\BIBdecl}{\relax}
\BIBdecl

\bibitem{Boas1990}
P.~van Emde~Boas, ``Machine models and simulation,'' in \emph{Handbook of
  Theoretical Computer Science. volume A: Algorithms and Complexity},
  J.~Van~Leeuwen, Ed.\hskip 1em plus 0.5em minus 0.4em\relax Elsevier, 1990,
  pp. 1--66.

\bibitem{DalLago2012a}
U.~Dal~Lago, ``A short introduction to implicit computational complexity,'' in
  \emph{ESSLLI}, ser. LNCS, N.~Bezhanishvili and V.~Goranko, Eds., vol.
  7388.\hskip 1em plus 0.5em minus 0.4em\relax Springer, 2011, pp. 89--109.

\bibitem{Bellantoni1992a}
S.~J. Bellantoni and S.~A. Cook, ``A new recursion-theoretic characterization
  of the polytime functions,'' \emph{Comput. Complex.}, vol.~2, pp. 97--110,
  1992.

\bibitem{Leivant1993}
D.~Leivant, ``Stratified functional programs and computational complexity,'' in
  \emph{POPL}, M.~S. Van~Deusen and B.~Lang, Eds.\hskip 1em plus 0.5em minus
  0.4em\relax ACM Press, 1993, pp. 325--333.

\bibitem{Neergaard2004}
P.~Neergaard, ``A functional language for logarithmic space,'' in \emph{APLAS},
  ser. LNCS, W.-N. Chin, Ed.\hskip 1em plus 0.5em minus 0.4em\relax Springer,
  2004, vol. 3302, pp. 311--326.

\bibitem{Girard1987}
J.-Y. Girard, ``Linear logic,'' \emph{Theoret. Comput. Sci.}, vol.~50, no.~1,
  pp. 1--101, 1987.

\bibitem{Girard1995}
------, ``Light linear logic,'' in \emph{LCC}, ser. LNCS, D.~Leivant, Ed.\hskip
  1em plus 0.5em minus 0.4em\relax Springer, 1995, vol. 960, pp. 145--176.

\bibitem{Danos2003}
V.~Danos and J.-B. Joinet, ``Linear logic \& elementary time,'' \emph{Inf.
  Comput.}, vol. 183, no.~1, pp. 123--137, 2003.

\bibitem{Lafont2004}
Y.~Lafont, ``Soft linear logic and polynomial time,'' \emph{Theoret. Comput.
  Sci.}, vol. 318, no.~1, pp. 163--180, 2004.

\bibitem{Baillot2004}
P.~Baillot and K.~Terui, ``Light types for polynomial time computation in
  lambda-calculus,'' in \emph{LICS}.\hskip 1em plus 0.5em minus 0.4em\relax
  IEEE Computer Society, 2004, pp. 266--275.

\bibitem{Lago2010d}
U.~Dal~Lago and U.~Schöpp, ``Functional programming in sublinear space,'' in
  \emph{ESOP}, ser. LNCS, A.~D. Gordon, Ed., vol. 6012.\hskip 1em plus 0.5em
  minus 0.4em\relax Springer, 2010, pp. 205--225.

\bibitem{Gaboardi2012}
M.~Gaboardi, J.-Y. Marion, and S.~Ronchi Della~Rocca, ``An implicit
  characterization of pspace,'' \emph{ACM Trans. Comput. Log.}, vol.~13, no.~2,
  pp. 18:1--18:36, 2012.

\bibitem{Girard1989b}
J.-Y. Girard, ``Towards a geometry of interaction,'' in \emph{Proceedings of
  the AMS-IMS-SIAM Joint Summer Research Conference held June 14-20, 1987},
  ser. Categories in Computer Science and Logic, J.~W. Gray and
  A.~{\v{S}}{\v{c}}edrov, Eds., vol.~92.\hskip 1em plus 0.5em minus 0.4em\relax
  AMS, 1989, pp. 69--108.

\bibitem{Danos1990}
V.~Danos, ``La logique linéaire appliquée à l’étude de divers processus
  de normalisation (principalement du \(\lambda\)-calcul),'' Ph.D.
  dissertation, Université Paris VII, 1990.

\bibitem{Seiller2012}
\BIBentryALTinterwordspacing
T.~Seiller, ``Logique dans le facteur hyperfini : géometrie de l'interaction
  et complexité,'' Ph.D. dissertation, Université de la Méditerranée, 2012.
  [Online]. Available: \url{http://tel.archives-ouvertes.fr/tel-00768403/}
\BIBentrySTDinterwordspacing

\bibitem{Seiller2014a}
\BIBentryALTinterwordspacing
------, ``Interaction graphs: Graphings,'' \emph{Submitted}, 2014. [Online].
  Available: \url{http://arxiv.org/pdf/1405.6331}
\BIBentrySTDinterwordspacing

\bibitem{Girard1989a}
J.-Y. Girard, ``Geometry of interaction 1: Interpretation of system {F},''
  \emph{Studies in Logic and the Foundations of Mathematics}, vol. 127, pp.
  221--260, 1989.

\bibitem{Girard2011a}
------, ``Geometry of interaction {V}: logic in the hyperfinite factor,''
  \emph{Theoret. Comput. Sci.}, vol. 412, no.~20, pp. 1860--1883, Apr. 2011.

\bibitem{Seiller2014b}
\BIBentryALTinterwordspacing
T.~Seiller, ``A correspondence between maximal abelian sub-algebras and linear
  logic fragments,'' \emph{Submitted}, 2014. [Online]. Available:
  \url{http://arxiv.org/pdf/1408.2125}
\BIBentrySTDinterwordspacing

\bibitem{Girard1995a}
J.-Y. Girard, ``Geometry of interaction {III}: accommodating the additives,''
  in \emph{Advances in Linear Logic}, ser. London Math. Soc. Lecture Note Ser.,
  J.-Y. Girard, Y.~Lafont, and L.~Regnier, Eds.\hskip 1em plus 0.5em minus
  0.4em\relax CUP, Jun. 1995, no. 222, pp. 329--389.

\bibitem{Bagnol2014}
M.~Bagnol, ``On the resolution semiring,'' Ph.D. dissertation, Aix-Marseille
  Université -- Institut de Mathématiques de Marseille, 2014.

\bibitem{Girard2012}
J.-Y. Girard, ``Normativity in logic,'' in \emph{Epistemology versus Ontology},
  ser. Logic, Epistemology, and the Unity of Science, P.~Dybjer, S.~Lindström,
  E.~Palmgren, and G.~Sundholm, Eds.\hskip 1em plus 0.5em minus 0.4em\relax
  Springer, 2012, vol.~27, pp. 243--263.

\bibitem{Aubert2014ctemp}
C.~Aubert and T.~Seiller, ``Characterizing co-nl by a group action,''
  \emph{MSCS (FirstView)}, pp. 1--33, 12 2014.

\bibitem{Aubert2015temp}
------, ``Logarithmic space and permutations,'' \emph{Inf. Comput.}, 2015, to
  appear.

\bibitem{Robinson1965}
J.~A. Robinson, ``A machine-oriented logic based on the resolution principle,''
  \emph{J. ACM}, vol.~12, no.~1, pp. 23--41, Jan. 1965.

\bibitem{Dantsin2001}
E.~Dantsin, T.~Eiter, G.~Gottlob, and A.~Voronkov, ``Complexity and expressive
  power of logic programming,'' \emph{ACM Comput. Surv.}, vol.~33, no.~3, pp.
  374--425, 2001.

\bibitem{Aubert2014}
C.~Aubert and M.~Bagnol, ``Unification and logarithmic space,'' in
  \emph{RTA-TLCA}, ser. LNCS, G.~Dowek, Ed., vol. 8650.\hskip 1em plus 0.5em
  minus 0.4em\relax Springer, 2014, pp. 77--92.

\bibitem{Aubert2014b}
C.~Aubert, M.~Bagnol, P.~Pistone, and T.~Seiller, ``Logic programming and
  logarithmic space,'' in \emph{APLAS}, ser. LNCS, J.~Garrigue, Ed., vol.
  8858.\hskip 1em plus 0.5em minus 0.4em\relax Springer, 2014, pp. 39--57.

\bibitem{Cook1971a}
S.~A. Cook, ``Linear time simulation of deterministic two-way pushdown
  automata,'' in \emph{IFIP Congress (1)}.\hskip 1em plus 0.5em minus
  0.4em\relax North-Holland, 1971, pp. 75--80.

\bibitem{Dwork1984}
C.~Dwork, P.~C. Kanellakis, and J.~C. Mitchell, ``On the sequential nature of
  unification,'' \emph{J. Log. Program.}, vol.~1, no.~1, pp. 35--50, 1984.

\bibitem{Baader1998a}
F.~Baader and T.~Nipkow, \emph{Term rewriting and all that.}\hskip 1em plus
  0.5em minus 0.4em\relax CUP, 1998.

\bibitem{Baillot2001}
P.~Baillot and M.~Pedicini, ``Elementary complexity and geometry of
  interaction,'' \emph{Fund. Inform.}, vol.~45, no. 1--2, pp. 1--31, 2001.

\bibitem{Ladermann1994}
M.~Ladermann and H.~Petersen, ``Notes on looping deterministic two-way pushdown
  automata,'' \emph{Inf. Process. Lett.}, vol.~49, no.~3, pp. 123--127, 1994.

\bibitem{Hartmanis1972}
J.~Hartmanis, ``On non-determinancy in simple computing devices,'' \emph{Acta
  Inform.}, vol.~1, no.~4, pp. 336--344, 1972.

\bibitem{Cook1971}
S.~A. Cook, ``Characterizations of pushdown machines in terms of time-bounded
  computers,'' \emph{J. ACM}, vol.~18, no.~1, pp. 4--18, 1971.

\bibitem{Wagner1986}
K.~W. Wagner and G.~Wechsung, \emph{Computational Complexity}, ser. Mathematics
  and its Applications.\hskip 1em plus 0.5em minus 0.4em\relax Springer, 1986,
  vol.~21.

\bibitem{Aho1968}
A.~V. Aho, J.~E. Hopcroft, and J.~D. Ullman, ``Time and tape complexity of
  pushdown automaton languages,'' \emph{Inform. Control}, vol.~13, no.~3, pp.
  186--206, 1968.

\bibitem{Michie1968}
D.~Michie, ``“{M}emo” functions and machine learning,'' \emph{Nature}, vol.
  218, pp. 19--22, Apr. 1968.

\bibitem{Amtoft1992}
T.~Amtoft and J.~L. Tr{\"a}ff, ``Partial memoization for obtaining linear time
  behavior of a 2dpda,'' \emph{Theoret. Comput. Sci.}, vol.~98, no.~2, pp.
  347--356, 1992.

\bibitem{Andersen1994}
N.~Andersen and N.~D. Jones, ``Generalizing cook's transformation to imperative
  stack programs,'' in \emph{Results and Trends in Theoretical Computer
  Science}, ser. LNCS, J.~Karhum{\"a}ki, H.~A. Maurer, and G.~Rozenberg, Eds.,
  vol. 812.\hskip 1em plus 0.5em minus 0.4em\relax Springer, 1994, pp. 1--18.

\bibitem{Gluck2013}
R.~Gl{\"u}ck, ``Simulation of two-way pushdown automata revisited,'' in
  \emph{Festschrift for Dave Schmidt}, ser. EPTCS, A.~Banerjee, O.~Danvy, K.-G.
  Doh, and J.~Hatcliff, Eds., vol. 129, 2013, pp. 250--258.

\bibitem{Ladner1975}
R.~E. Ladner, ``The circuit value problem is log space complete for p,''
  \emph{ACM SIGACT News}, vol.~7, no.~1, pp. 18--20, 1975.

\bibitem{Danos1999}
V.~Danos and L.~Regnier, ``Reversible, irreversible and optimal
  lambda-machines,'' \emph{Theoret. Comput. Sci.}, vol. 227, no. 1-2, pp.
  79--97, 1999.

\end{thebibliography}

\vspace*{-1em}
\end{document}